\documentclass[a4paper,UKenglish, autoref, thm-restate]{lipics-v2021} 
\usepackage{thm-restate}
\usepackage{cite}

\newcommand{\cand}{\mathsf{cand}}
\newcommand{\sums}{\mathsf{sums}}
\newcommand{\ch}{\mathsf{ch}}
\newcommand{\argmin}{\mathsf{argmin}}

\newcommand{\cA}{\mathcal{A}}
\newcommand{\cB}{\mathcal{B}}
\newcommand{\cL}{\mathcal{L}}
\newcommand{\cP}{\mathcal{P}}
\newcommand{\cR}{\mathcal{R}}
\newcommand{\cS}{\mathcal{S}}
\newcommand{\cT}{\mathcal{T}}
\newcommand{\cW}{\mathcal{W}}
\newcommand{\cX}{\mathcal{X}}
\newcommand{\cY}{\mathcal{Y}}
\newcommand{\A}{\mathbf{A}}
\newcommand{\U}{\mathbb{U}}
\newcommand{\ind}{\mathsf{ind}}

\newcommand{\pref}{\mathsf{pref}}
\newcommand{\suff}{\mathsf{suff}}

\title{On Minimizers of Minimum Density}

\author{Arseny Shur}{Bar Ilan University, Israel}{shur@datalab.cs.biu.ac.il}{https://orcid.org/0000-0002-7812-3399}{
}
\authorrunning{A. Shur} 

\Copyright{Arseny Shur} 



\nolinenumbers

\begin{document}

\maketitle

\begin{abstract}
Minimizers are sampling schemes with numerous applications in computational biology.
Assuming a fixed alphabet of size $\sigma$, a minimizer is defined by two integers $k,w\ge2$ and a linear order $\rho$ on strings of length $k$ (also called $k$-mers).
A string is processed by a sliding window algorithm that chooses, in each window of length $w+k-1$, its minimal $k$-mer with respect to $\rho$.
A key characteristic of the minimizer is its density, which is the expected frequency of chosen $k$-mers among all $k$-mers in a random infinite $\sigma$-ary string.
Minimizers of smaller density are preferred as they produce smaller samples with the same guarantee: each window is represented by a $k$-mer. 

The problem of finding a minimizer of minimum density for given input parameters $(\sigma,k,w)$ has a huge search space of $(\sigma^k)!$ and is representable by an ILP of size $\tilde\Theta(\sigma^{k+w})$, which has worst-case solution time that is doubly-exponential in $(k+w)$  under standard complexity assumptions.
We solve this problem in $w\cdot 2^{\sigma^k+O(k)}$ time and provide several additional tricks reducing the practical runtime and search space.
As a by-product, we describe an algorithm computing the average density of a minimizer within the same time bound.
Then we propose a novel method of studying minimizers via regular languages and show how to find, via the eigenvalue/eigenvector analysis over finite automata, minimizers with the minimal density in the asymptotic case $w\to\infty$. 
Implementing our algorithms, we compute the minimum density minimizers for $(\sigma,k)\in\{(2,2),(2,3),(2,4),(2,5),(4,2)\}$ and \textbf{all} $w\ge 2$.
The obtained densities are compared against the average density and the theoretical lower bounds, including the new bound presented in this paper.
\end{abstract}
\clearpage

\section{Introduction}

Sampling short substrings ($k$-mers) in long biological sequences is an important step in solving many bioinformatics tasks.
Typically, a ``window guarantee'' is required: each window of $w$ consecutive $k$-mers (i.e., of length $w+k-1$) in the input sequence should be represented by at least one $k$-mer in the sample.
Minimizers, introduced in \cite{SWA03,RHHMY04}, are simple sampling schemes with the window guarantee: a linear order on $k$-mers is fixed, and in each window, the minimal $k$-mer w.r.t. this order is selected, with ties broken to the left.
A comprehensive list of bioinformatics applications using minimizers can be found in \cite{Ndi2024survey}.
Note that some other sampling schemes use minimizers as an intermediate step \cite{Edg21,KP24,Sah21}.

Minimizers are characterized by their density, which is the expected fraction of sampled $k$-mers in an infinite sequence of i.i.d. symbols.
The average density (a.k.a. expected density of a random minimizer) is close to $\frac{2}{w+1}$ unless $w\gg k$ \cite{SWA03}; for more details see \cite{GoSh25}.
Minimizers of lower density are desirable, as they produce smaller samples with the same window guarantee.
Many methods were designed to build low-density minimizers.
Some of them are constructions, like miniception \cite{zheng2020miniception}, double-decycling \cite{PPEKBSO2023decycling}, and open-closed syncmers \cite{OPMSK2017DOCS}, while others are search algorithms, like DOCKS \cite{OPMSK2017DOCS}, PASHA \cite{EBO2020PASHA}, and GreedyMini \cite{GTKOS25greedymini}.
However, one can only conjecture how close their results are to the minimum density, as the exact values of the latter are not known apart from few small particular cases.

There are two trivial lower bounds on minimizers density: $\frac{1}{w}$ (due to window guarantee) and $\frac{1}{\sigma^k}$ (all occurrences of the minimal $k$-mer are sampled).
Both are ``weakly'' reachable \cite{MDK18asymptotic}: for fixed $w$, there is an infinite sequence of minimizers with densities approaching $\frac{1}{w}$ as $k\to\infty$; for fixed $k$, infinitely growing $w$ and \emph{every} order, the density approaches $\frac{1}{\sigma^k}$. 
The recent lower bound $\max\big\{\frac{1}{w+k}\big\lceil\frac{w+k}{w}\big\rceil, \frac{1}{w+k'}\big\lceil\frac{w+k'}{w}\big\rceil\big\}$, where $k'=\big\lceil\frac{k-1}{w}\big\rceil w+1$, \cite{KKMLT24lower} works for a wider class of \emph{forward} sampling schemes, but for $w\le k$ is nearly tight for minimizers also.
This bound is hit by minimizers in several points \cite{GTKOS25greedymini}.
However, apart from few hits, there is still a gap in density between the lower bounds and the best known minimizers.

In this paper, we study minimum densities and the minimizers reaching these minima.
The size $\sigma$ of the alphabet is assumed constant.
In Section~\ref{s:exhaust}, we present an algorithm finding the minimizer of minimum density for the triple $(\sigma,k,w)$ in time linear in $w$ and doubly exponential in $k$. 
This is a significant improvement over an ILP of size $\tilde{O}(\sigma^{k+w})$ \cite{KKMLT24lower}, as confirmed by the computational results in Section~\ref{s:experiments}.
As a byproduct, we give an algorithm for average density with the same working time.
In Section~\ref{s:asymp}, we develop a novel method of studying minimizers through regular languages.
We show how to construct orders that generate minimum density minimizers for infinitely many window sizes, then build such orders for several small $(\sigma,k)$ pairs, and improve the density lower bound for big $w$.
Finally, in Section~\ref{s:experiments} we present and discuss our computational results: the minimum densities for $(\sigma,k)\in\{(2,2),(2,3),(2,4),(2,5),(4,2)\}$ and \emph{all} $w$.
Appendix~\ref{ss:proofs} contains omitted proofs.



\section{Preliminaries}

In what follows, $\Sigma, \sigma,k$, and $w$ denote, respectively, the alphabet $\{0,\ldots,\sigma{-}1\}$, its size, the length of the sampled substrings ($k$-mers) and the number of $k$-mers in a window.
We write $s[1..n]$ for a length-$n$ string, denote the length of $s$ by $|s|$, and use standard definitions of substring, prefix, and suffix of a string.
We write $n$-string ($n$-prefix, $n$-suffix, $n$-window) to indicate length.
We use the notation $[i..j]$ for the range of integers from $i$ to $j$, and $s[i..j]$ for the substring of $s$ covering the range $[i..j]$ of positions.
A string $s$ \emph{avoids} a string $t$ is $t$ is not a substring of $s$.
An integer $p<|s|$ is a \emph{period} of $s$ if $s[1..|s|-p]=s[p{+}1..|s|]$.
For a string $s$ and a rational $\alpha>1$, we write $s^\alpha$ for the $(\alpha|s|)$-prefix of the infinite string $sss\cdots$. 

We consider deterministic finite automata (DFAs) with \emph{partial} transition function: for a state $q$ and letter $a$, the transition from $q$ by $a$ can be missing. 
We view DFAs as labelled digraphs.
For a fixed DFA, we write $q.w$ to denote the state reached from the state $q$ by a walk labelled by $w$.
We write $\tilde{O}(f)$ to suppress $\mathrm{polylog}(f)$ factors.

Let $\pi$ be a permutation ($=$ a linear order) of $\Sigma^k$.
We view $\pi$ as a bijection $\pi:[1..\sigma^k]\to\Sigma^k$, binding all $k$-mers to their \emph{$\pi$-ranks}.
The \emph{minimizer} $(\pi,w)$ is a map $f:\Sigma^{w+k-1}\to [1..w]$ assigning to each $(w{+}k{-}1)$-window the starting position of its minimum-$\pi$-rank $k$-mer, with ties broken to the left.
This map acts on strings over $\Sigma$, selecting one position in each window so that in a window $v=S[i..i{+}w{+}k{-}2]$ the position $i+f(v)-1$ in $S$ is selected.
Minimizers are evaluated by the density of selected positions.
Let $f(S)$ denote the set of positions selected in a string $S$ by a minimizer $f$. 
The \emph{density of $f$} is the limit $d_{f}=\lim_{n\to \infty}\frac{1}{\sigma^n}\sum_{S\in\Sigma^n} d_{f}(S)$.

An \emph{arrangement} of $\Sigma^k$ with the \emph{domain} $\varnothing\ne U\subseteq\Sigma^k$ is an arbitrary permutation of $U$.
We view $\pi$ as a string of length $|\pi|=|U|$.
We write $\pi_1\cdot\pi_2$ for the arrangement that is the concatenation of arrangements $\pi_1$ and $\pi_2$.
Arrangements of length $\sigma^k$ are exactly (linear) orders.
We extend the notion of $\pi$-rank to arbitrary arrangement $\pi$.
The \emph{$\pi$-minimal} $k$-mer in a string (e.g., in a window) is its $k$-mer of minimum $\pi$-rank.

A subset $H\subseteq \Sigma^k$ is a \emph{universal hitting set} (UHS) \emph{for $w$} if every $(w{+}k{-}1)$-window contains at least one $k$-mer from $H$. 
For example, $\{00,01,11\}$ is a UHS for $\sigma=k=2$ and any $w>1$.
We call an arrangement $\pi$ a \emph{UHS order for $w$} if its domain is a UHS for $w$.
Then any two orders $\pi\cdot\pi_1$ and $\pi\cdot\pi_2$ define the same minimizer for $w$, because the minimum-rank $k$-mer in every window is in $\pi$; we denote this minimizer by $(\pi,w)$.
We write $\U_{\sigma,k}$ for the set of all arrangements of $\Sigma^k$ that are UHS orders for some $w$ (and hence for all sufficiently big $w$).

For a given minimizer $f=(\rho,w)$, a $(w{+}k)$-window $v$ (which contains two consecutive $(w{+}k{-}1)$-windows) is \emph{charged} if its minimum-rank $k$-mer is either its prefix or its \emph{unique suffix} (i.e., the $k$-suffix of $v$ having no other occurrence in $v$); otherwise, $v$ is \emph{free}.
An important observation is that every string $S$ contains exactly $|f(S)|-1$ (not necessarily distinct) charged $(w{+}k)$-windows \cite[Lemma 6]{ZMK2023sketches}. 
Since all possible $n$-strings have, in total, the same number of occurrences of each $(w{+}k)$-window, the density $d_{f}$ of a minimizer equals the fraction of charged windows in $\Sigma^{w+k}$~\cite{MPBOSK2017improving,zheng2020miniception}. 
For fixed $k,w$, ``window'' means a $(w{+}k)$-window.

For an  arrangement $\pi$ of $\Sigma^k$, a window $v$ is \emph{charged by $\pi$ due to $\pi[i]$} if the $\pi$-minimal $k$-mer of $v$ is $\pi[i]$ and is a prefix (\emph{prefix-charged}) or a unique suffix (\emph{suffix-charged}) of $v$.
A window is \emph{live} (w.r.t. $\pi$) if it has no $k$-mers in $\pi$.
A UHS order $\rho\in\mathbb{U}_{\sigma,k}$ is \emph{optimal for $w$}, if the minimizer $(\rho,w)$ has the minimum density among all $(\sigma,k,w)$-minimizers.
Furthermore, $\rho\in\mathbb{U}_{\sigma,k}$ is \emph{eventually optimal} if it is optimal for infinitely many values of $w$.

If $\rho=(u_1,\ldots, u_s)$ is a UHS order and $h$ is a permutation of $\Sigma$, then trivially $h(\rho)=(h(u_1),\ldots,h(u_s))$  is a UHS order with exactly the same characteristics.
Namely, any window $v$ is prefix-charged by $\rho$ due to $\rho[i]$ if and only if $h(v)$ is prefix-charged by $h(\rho)$ due to $h(\rho)[i]$; the same property holds for suffix-charged windows.
In particular, $d_{(\rho,w)}=d_{(h(\rho),w)}$ for all $w$.
Due to this symmetry, we focus on \emph{lexmin} orders, which are lexicographically smaller than their image under any permutation of $\Sigma$.

\section{Computing optimal minimizers} \label{s:exhaust}

In \cite[Suppl. C]{KKMLT24lower}, an integer linear program of size  $\Theta(\sigma^{k+w})$ was given to define a forward scheme of minimum density.
With all improvements over the basic ILP, the authors of \cite{KKMLT24lower} were able to find the minimum-density schemes for only a few non-trivial $(\sigma,k,w)$ triples.
Their program can be adjusted to get a $\tilde{\Theta}(\sigma^{k+w})$-size ILP defining the minimum-density minimizer.
This means that the solution is worst-case doubly-exponential in $k+w$.  
We present a more efficient search, utilizing a combinatorial property of minimum-density orders.

Suppose $\sigma,k$, and $w$ are fixed, $n=k+w$ is the window size.
With every arrangement $\pi$ of $\Sigma^k$ we associate two sequences of sets $\{\cX_{\pi,i}\}_1^{|\pi|}$ and $\{\cY_{\pi,i}\}_1^{|\pi|}$ such that $\cX_{\pi,i}\subset \Sigma^n$ consists of all $n$-windows with the $\pi$-minimal $k$-mer $\pi[i]$, and $\cY_{\pi,i}\subseteq \cX_{\pi,i}$ is the set of all $n$-windows charged by $\pi$ due to $\pi[i]$.
Let $\ch(\pi,i)=|\cY_{\pi,i}|$; the number of windows charged by $\pi$ is $\ch(\pi)=\sum_{i=1}^{|\pi|} \ch(\pi,i)$.
The following lemma is immediate.
\begin{lemma} \label{l:locality}
    Given an arrangement $\pi$, let $\pi'$ be an arrangement obtained from $\pi$ by permuting some $k$-mers inside the range $[j_1..j_2]$ of indices.
    Then $\cX_{\pi,i}=\cX_{\pi',i}$ and $\cY_{\pi,i}=\cY_{\pi',i}$ for every $i\in[1..j_1-1]\cup[j_2+1..|\pi|]$.
\end{lemma}

An arrangement $\pi$ with the domain $U$ is \emph{optimal} if $\ch(\pi)\le \ch(\pi')$ for every arrangement $\pi'$ of $U$.
By Lemma~\ref{l:locality} and the definition of $\ch(\pi)$, we have
\begin{lemma} \label{l:opt_arrange}
    Every prefix of an optimal arrangement is optimal.
\end{lemma}

To present the main result of this section, we need the following lemma, implicit in \cite{GTKOS25greedymini}.
\begin{restatable}{lemma}{densitylemma} \label{l:ch}
    For given $\sigma$, $k$, $w$, an arrangement $\pi$ on $\Sigma^k$, and an index $i\in[1..|\pi|]$, the number $\ch(\pi,i)$ can be computed in time $O(\min\{\sigma^w,w\sigma^k\})$.
\end{restatable}

\begin{theorem} \label{t:exhaustive}
    For given $\sigma$, $k$, and $w$, a minimizer $(\rho,w)$ of minimum density can be found in $O(T{\cdot} 2^{\sigma^k})$ time and $O(\sigma^{k/2}{\cdot} 2^{\sigma^k})$ space, where $T=\min\{\sigma^{w+k},w\sigma^{2k}\}$.
\end{theorem}
\begin{proof}
    The algorithm computing the order $\rho$ proceeds in phases.
    At the end of $t$'th phase, it stores $\tbinom{\sigma^k}{t}$ key-value pairs.
    The keys are all $t$-element subsets of $\Sigma^k$, represented by $\sigma^k$-bit masks; the value of the key $U$ is the pair $(\pi,\ch(\pi))$, where $\pi$ is an optimal arrangement with the domain $U$.
    Thus, at the end of the last phase the only stored value contains an optimal order $\rho$.
    By definition of optimal, $\rho$ has the minimum density (which is $\frac{\ch(\rho)}{\sigma^{w+k}}$).

    Let $\pi$ be an optimal arrangement with the domain $U=\{u_1,\ldots,u_t\}\subseteq \Sigma^k$, and let $\pi[t]=u_i$.
    By Lemma~\ref{l:opt_arrange}, $\pi=\pi_i\cdot(u_i)$, where
    $\pi_i$ is an optimal arrangement with the domain $U\setminus\{u_i\}$.
    Accordingly, the algorithm computes an optimal arrangement with the domain $U$ as $\pi=\argmin\{\ch(\pi_i\cdot(u_i))\mid i=1,\ldots,t\}$, where  $\pi_i$ is optimal with the domain $U\setminus\{u_i\}$.
    
    Let us describe the details. 
    In the first phase the algorithm processes all 1-element arrangements $\pi=(u)$, $u\in\Sigma^k$, computing $\ch(\pi)=|\cY_{\pi,1}|$ and storing the value $(\pi,\ch(\pi))$ by the key $\{u\}$. 
    During the $(t+1)$th phase, where $t\ge 1$, it loops over all keys, which are $t$-element sets of $k$-mers.
    Given a key $U$ with the value $(\pi,\ch(\pi))$, it creates the arrangement $\pi_u=\pi\cdot(u)$ for every $k$-mer $u\notin U$, computes $\ch(\pi_u)=\ch(\pi)+\ch(\pi_u,t{+}1)$, and looks up  $U'=U\cup\{u\}$ in the dictionary. 
    If no entry for $U'$ exists, it is created with the value $(\pi_u, \ch(\pi_u))$.
    If the entry exists (i.e., it was created earlier during this phase when processing another $t$-element subset of $U'$), its value $(\tau,\ch(\tau))$ is replaced with $(\pi_u, \ch(\pi_u))$ if $\ch(\pi_u)<\ch(\tau)$.
    Hence, when all $t$-element subsets of $U'$ are processed, the value by the key $U'$ contains an optimal arrangement with the domain $U'$.
    After processing all $u\notin U$, the key $U$ is deleted.
    Therefore, at the end of the $(t+1)$th phase the keys in the dictionary are exactly the $(t+1)$-element subsets of $\Sigma^k$.
    This means that at any moment the keys are subsets of $\Sigma^k$ of two consecutive sizes.
    Then the number of dictionary entries is $O(2^{\sigma^k}\!/\sigma^{k/2})$, with each entry of size $O(\sigma^k)$.
    Thus the size of the dictionary fits into the theorem's space bound.

    For each subset $U$ of $\Sigma^k$, the algorithm computes $O(\sigma^k)$ numbers $\ch(\pi_u,|U|+1)$, where $\pi_u$ is defined above; the first phase corresponds to $U=\varnothing$.
    By Lemma~\ref{l:ch}, such a number can be computed in $O(\min\{\sigma^w,w\sigma^k\})$ time.
    Multiplying this by the number of computations, we get the theorem's time bound.
    The time spent for dictionary calls and the space used by a computation of $\ch(\cdot)$ are both negligible compared to the already computed bounds.
\end{proof}
Since the algorithm of Theorem~\ref{t:exhaustive} is linear in $w$, it can be used for almost any $w$ but only very small $k$.
Below we describe a few improvements that greatly decrease both the time and space used.
E. g., the case of $\sigma=2$, $k=5$ becomes feasible for an ordinary laptop.
See Section~\ref{s:experiments} for the computational results.

\noindent \textbf{1. Using an upper bound.} 
Given an order $\bar\rho$, drop every arrangement $\pi$ such that $\ch(\pi)\ge \ch(\bar\rho)$ from further processing.
Thus, if no arrangement $\pi$ with the domain $U$ satisfies $\ch(\pi)<\ch(\bar\rho)$, then $U$ never appears in the dictionary.
This simple trick decreases the number of subsets added to the dictionary, and reduces the runtime accordingly.
A further reduction in runtime is achieved due to the fact that each of the algorithms of Lemma~\ref{l:ch} consists of two consecutive subroutines to count prefix-charged and suffix-charged windows. 
Comparing the number of charged windows to $\ch(\bar\rho)$ after the first subroutine sometimes allows one to drop the current arrangement $\pi$ without running the second subroutine.

The quality of search reduction depends on the proximity of $\bar\rho$ to the optimal density.
In our computations, we took the best order obtained by the greedy method \cite{GTKOS25greedymini} for each $w\le 15$; otherwise we took the order from the optimal minimizer $(\rho,w-1)$ computed before.

\noindent \textbf{2. Using UHSs.}
If $\pi$ is a UHS order, then $\bigcup_{i=1}^{|\pi|}\cX_{\pi,i}=\Sigma^{k+w}$.
Note that the domain $U$ of $\pi$ is a UHS if and only if $\ch(\pi\cdot(u))=\ch(\pi)$ for every $u\notin U$; this condition can be checked for free when processing $U$.
If the domain of $\pi$ is a UHS, then $\ch(\rho)=\ch(\pi)$ for every order $\rho$ with the prefix $\pi$.
This equality has two implications: $\ch(\pi)$ is a new upper bound for the remaining search (it is smaller than the previous bound because $U$ appeared in the dictionary) and all arrangements with the prefix $\pi$ can be dropped from the search.

The above UHS trick suffers from the property that a UHS must contain all unary $k$-mers, because every unary window is charged due to its unique $k$-mer.
We enhance this trick as follows.
For an arrangement $\pi$ with the domain $U$, consider $\ch'(\pi)=\ch(\pi)+|\{a\mid a^k\notin U\}|$.
Let $\rho$ be an order with the prefix $\pi$.
Then $\ch(\rho)\ge \ch'(\pi)$, as each unary $k$-mer charges at least one string.
Now if $\ch'(\pi\cdot(u))=\ch'(\pi)$ for all $u\notin U$, then all windows having no $k$-mer in $U$ are unary, implying $\ch(\rho)=\ch'(\pi)$.
Hence $\ch'(\pi)$ is a new upper bound, and all arrangements with the prefix $\pi$ can be dropped.

\noindent \textbf{3. Reducing space.}
As mentioned in the preliminaries, we can consider only lexmin arrangements.
However, this only marginally effects the storage.
E. g., a subset $U\in\{0,1\}^k$ is guaranteed to not appear in the dictionary only if all $k$-mers in $U$ start with 1.

A much better trick is to store in the dictionary the values $\ch(\pi)$ instead of $(\pi,\ch(\pi))$.
Indeed, the arrangement $\pi$ is not needed inside the algorithm: due to Lemma~\ref{l:locality}, the algorithms of Lemma~\ref{l:ch} can take any arrangement with the domain $U$.
Thus, the algorithm will return $\ch(\rho)$ for an optimal order $\rho$.
To restore $\rho$, store the set $U$ that gave you the last upper bound (i.e., $\ch(\rho)=\ch'(\pi)$ for an optimal arrangement $\pi$ with the domain $U$).
Next we run the search version that \emph{stores} the arrangements, but only those picked from the set $U$ instead of $\Sigma^k$.
In our experiments for the critical case $\sigma=2, k=5$, the size of $U$ was never above 18 (out of $|\Sigma^k|=32$), so the second run took negligible time and space.

\smallskip
A simplified version of the algorithm from Theorem~\ref{t:exhaustive} can solve another density problem for minimizers.
Given $\sigma,k$, and $w$, the \emph{average density} is the average of densities of all minimizers with the parameters $(\sigma,k,w)$.
In the case $w\le k$, it can be computed in $O(w\log w)$ time due to a special structure of windows containing repeated $k$-mers \cite{GoSh25}, but for $w>k$ no algorithms that are polynomial in $k$ and/or $w$ were known.
We show the following result.

\begin{restatable}{theorem}{average} \label{t:random}
    The average density $\cR_\sigma(k,w)$ of a minimizer with the parameters $(\sigma,k,w)$ can be computed in $O(w2^{\sigma^k+2k})$ time and $O(\sigma^k)$ space.
\end{restatable}

\section{Asymptotic comparison of orders} \label{s:asymp}

For $\rho,\rho'\in\U_{\sigma,k}$ we say that $\rho$ is \emph{asymptotically better} than $\rho'$, denoted by $\rho\triangleleft \rho'$, if there exists $w_0$ such  that $w\ge w_0$ implies $d_{(\rho,w)}< d_{(\rho',w)}$.
We extend this notion to arbitrary arrangements: given two arrangements $\pi,\pi'$ of $\Sigma^k$, we write $\pi\triangleleft\pi'$ if $\rho\triangleleft \rho'$ for all $\rho,\rho'\in\U_{\sigma,k}$ such that $\pi$ is a prefix of $\rho$ and $\pi'$ is a prefix of $\rho'$.
The tools for order comparison are presented in Section~\ref{ss:regcharged} after describing necessary results on regular languages in Section~\ref{ss:growth}.

The search for eventually optimal UHS orders is organized as follows.
Due to symmetry, we process only lexmin arrangements.
At $i$'th iteration, we build a list $\cand_i$ of ``candidate'' arrangements of size $i$.
We make this list as small as possible while preserving the main property: $\cand_i$ contains $i$-prefixes of all eventually optimal lexmin orders, up to trivial permutations described in Section~\ref{ss:optimal}.
The choice of candidates is governed by a simple rule: if \emph{we can prove} $\pi\triangleleft\pi'$ for the arrangements $\pi$ and $\pi'$ of size $i$, then $\pi'\notin \cand_i$.
If for some $i$ there is a UHS order $\rho\in\cand_i$ such that \emph{we can disprove} the condition $\pi\triangleleft\rho$ for each $\pi\in\cand_i$, then we conclude that $\rho$ is eventually optimal.
Using this strategy, we find in Section~\ref{ss:optimal} the unique, up to symmetry and trivial permutations, eventually optimal orders for $\sigma=2$ and $k=2,3,4,5$, and also for $\sigma=4$ and $k=2$.
Finally, in Section~\ref{ss:lower} we discuss the density lower bound for big $w$.

\subsection{Growth functions of regular languages} \label{ss:growth}

We need some notation and facts on the growth of regular languages.
For any language $L\subseteq\Sigma^*$, its \emph{growth function} (a.k.a. combinatorial complexity) is the function counting strings: $C_L(n)=|L\cap\Sigma^n|$.
Its main asymptotic characteristic is the \emph{growth rate} $g(L)=\limsup_{n\to\infty}(C_L(n))^{1/n}$; if $L$ is closed under taking substrings, then $\limsup$ can be replaced by $\lim$.
If $L$ is regular, there exist efficient algorithms to calculate $g(L)$ within any prescribed error and to determine other asymptotic characteristics of $C_L(n)$ \cite{Shur10tcs,Shur10imm}.
The general form of $C_L(n)$ is given by the following theorem, combined from results of Salomaa and Soittola \cite{SaSo79}.

\begin{theorem}[\cite{SaSo79}] \label{t:SaSo}
For every regular language $L$ there exist $n_0,r\in\mathbb{N}$ such that for every $j\in[0..r-1]$ and every $n\ge n_0$ satisfying $n\bmod r = j$, the growth function of $L$ equals either 0 or the following finite sum:
\begin{equation} \label{e:SaSo}
    C_L(n)=p_j(n)\alpha_j^n+\sum\nolimits_i p_{ij}(n)\alpha_{ij}^n,
\end{equation}
where $\alpha_j\in\mathbb{R}^+$, $\alpha_{ij}\in\mathbb{C}$  are algebraic numbers, $\alpha_j>\max|\alpha_{ij}|$, $p_j\in\mathbb{R}[x]$ and $p_{ij}\in\mathbb{C}[x]$ are polynomials with algebraic coefficients.
\end{theorem}

As the strings in a regular language $L$ are in 1-to-1 correspondence with the accepting walks in any DFA $\cA$ accepting $L$, one can count walks instead of strings.
We recall some facts about counting walks in digraphs (see \cite{CDS95}).
Given a digraph (in particular, a DFA) $G$, its characteristic polynomial $\chi_G(r)$ and its eigenvalues are those of the adjacency matrix $A$ of $G$.
The number $W_{uv}(n)$ of $(u,v)$-walks of length $n$ in $G$ equals the entry $A^n[u,v]$ of the $n$'th power of $A$.
Hence this number $W_{uv}(n)$, as well as the total number $W_{**}(n)$ of length-$n$ walks in $G$ and the number $W_{u*}$ of such walks with a fixed initial vertex $u$, satisfies the homogeneous linear recurrence relation with the characteristic polynomial equal to $\chi_G(r)$. 
In particular, all non-zero numbers $\alpha_j,\alpha_{ij}$ in \eqref{e:SaSo} are the roots of $\chi_{\cA}(r)$, where $\cA$ accepts $L$.
The \emph{index} $\ind(G)$ is the \emph{dominant} eigenvalue of $G$, i.e, a positive eigenvalue $\alpha$ such that $\alpha\ge|\gamma|$ for any other eigenvalue $\gamma\in\mathbb{C}$ of $G$.
The dominant eigenvalue exists by the Perron--Frobenius theorem and  determines the asymptotic growth of the entries of the matrices $A^n$.
We say that $\alpha$ is \emph{strictly dominant} if its multiplicity is 1 and $\alpha>|\gamma|$ for every eigenvalue $\gamma\ne\alpha$.
For the case where $\ind(G)$ is strictly dominant, the following lemma provides quite precise asymptotic estimates of the functions $W_{uv}(n)$.

\begin{restatable}{lemma}{limitmatrix} \label{l:limit}
    Let $G$ be a digraph with the adjacency matrix $A$, let $\alpha=\ind(G)>1$ be strictly dominant, and let $\gamma$ be the second largest in absolute value eigenvalue of $G$.
    Then\\
    \emph{(i)} there exists a limit matrix $\A=\lim_{n\to\infty} (\alpha^{-1}A)^n=\vec{x}\vec{y}^\top$, where $\vec{x}$ and $\vec{y}^\top$ are, respectively, the column eigenvector and the row eigenvector of $A$ for the eigenvalue $\alpha$;\\
    \emph{(ii)} for every pair $(u,v)$ of vertices in $G$, $W_{uv}(n)=\A[u,v]\cdot\alpha^n+\tilde{O}(|\gamma|^n)$.
\end{restatable}

We call a digraph $G$ \emph{simple} if it consists of zero or more trivial strong components (singletons and cycles) and one nontrivial strong component $H$, and the greatest common divisor of lengths of all cycles in $H$ is $1$.
We call $H$ the \emph{major} component of $G$.
By the Perron--Frobenius theorem, the condition on components implies that $\ind(G)=\ind(H)>1$ is a simple eigenvalue, while the condition on cycle lengths guarantees that all other eigenvalues of $G$ have smaller absolute values than $\ind(G)$.
Therefore, Lemma~\ref{l:limit} holds for simple digraphs.

We call a digraph \emph{flat} if all its strong components are trivial.
A flat digraph has $\mathsf{poly}(n)$ walks of any length $n$.
All instances of DFAs that appear in this study are simple or flat.

\subsection{Regular languages of charged windows}
\label{ss:regcharged}

For any finite set $\varnothing\ne M\subset\Sigma^+$ of nonempty strings (\emph{antidictionary}), we define the language $L_M\subset\Sigma^*$ of all strings avoiding every string from $M$.
The language $L_M=\Sigma^*-\Sigma^*M\Sigma^*$ is regular; such languages are widely used in symbolic dynamics (subshifts of finite type), in algebra, and to study the growth of more complicated languages (see, e.g., \cite{Shur10tcs}).
In \cite{CMR98}, a modification of the Aho--Corasick automaton was proposed to accept languages with finite antidictionaries.
Given an antidictionary $M$, a DFA $\cA_M$, accepting $L_M$, is built as follows:
\begin{enumerate}
    \item create a trie $\cT$ for $M$, name each vertex by the label of the path from the root to it;
    \item add edges to get a complete DFA: an edge $u\xrightarrow{a}v$ is added if $v$ is the longest suffix of $ua$ that is a vertex of $\cT$;
    \item delete all leaves of $\cT$, mark $\lambda$ as the initial state, mark all vertices as terminal states.
\end{enumerate}
We refer to $\cA_M$ as the \emph{canonical} DFA for $L_M$.
It is \emph{trim} (each state is reachable from the initial state) but not necessarily minimal.
Every walk in $\cA_M$ is labelled by a string from $L_M$, and if $u$ is the label of a walk to the vertex $v$, then one of $u,v$ is the suffix of the other.

\smallskip
Now we define regular languages related to charged windows.
For every $\rho\in\U_{\sigma,k}$ and each $i\in[1..|\rho|]$, let $\cL_{\rho,i}$ be the language with the antidictionary $\{\rho[1],\ldots,\rho[i]\}$, $\cA_{\rho,i}$ be the canonical DFA accepting $\cL_{\rho,i}$, and $\alpha_{\rho,i}=g(\cL_{\rho,i})$.
We also formally define $\cL_{\rho,0}=\Sigma^*$.
Though all the automata $\cA_{\rho,i}$ that arise in our studies are either simple or flat, in general they can be much more complicated; some examples can be found in \cite[Section 6]{Shur13}.

Let $\cP_{\rho,i}$ be the language of all windows prefix-charged by $\rho$ due to the $k$-mer $\rho[i]$.
That is, $\cP_{\rho,i}\cap \Sigma^{w+k}$ is the set of all $(w+k)$-windows prefix-charged by the minimizer $(\rho,w)$ due to $\rho[i]$.
Thus, the growth function $C_{\cP_{\rho,i}}(n)$ counts such charged windows.
Note that $C_{\cP_{\rho,1}}(w+k)=\sigma^w$ for any $\rho$ and $w$ (all windows with the $k$-prefix $\rho[1]$ are charged due to it).
In a symmetric way, we define the languages $\cS_{\rho,i}$ of suffix-charged windows.
Finally, let $\cW_\rho=\bigcup_{i=1}^{|\rho|}\cP_{\rho,i}\cup \bigcup_{i=1}^{|\rho|}\cS_{\rho,i}$ be the language of all windows charged by $\rho$.
As the union is disjoint, we have $C_{\cW_\rho}(n)=\sum_{i=1}^{|\rho|}\big(C_{\cP_{\rho,i}}(n)+C_{\cS_{\rho,i}}(n)\big)=\ch(\rho)$.
Note that by definition of density $C_{\cW_\rho}(w+k)=d_{(\rho,w)}\sigma^{w+k}$ .
The following lemma proves regularity of the languages $\cP_{\rho,i}, \cS_{\rho,i}$, and expresses their growth functions in terms of the automata $\cA_{\rho,i}$.
\begin{lemma} \label{l:regular}
    Let $\rho\in\U_{\sigma,k}$, $i\in[1..|\rho|]$, and $\rho[i]=u=va$, where $a\in\Sigma$. 
    Then \\
    \emph{(i)} $\cP_{\rho,i}=u\Sigma^*\cap \cL_{\rho,i-1}$ and $\cS_{\rho,i}=(\cL_{\rho,i}\cap \Sigma^* v)a$; in particular, $\cP_{\rho,i}$ and $\cS_{\rho,i}$ are regular;\\
    \emph{(ii)} $C_{\cP_{\rho,i}}(n)$ is the number of walks of length $n-k$, starting at the vertex $\lambda.u$, in $\cA_{\rho,i-1}$;\\
    \emph{(iii)} $C_{\cS_{\rho,i}}(n)$ is the number of $(\lambda,v)$-walks of length $n-1$ in $\cA_{\rho,i}$.
\end{lemma}

\begin{proof}
    By definition, $\cP_{\rho,i}$ consists of all strings that have the $k$-prefix $u$ and avoid the set $\{\rho[1],\ldots,\rho[i-1]\}$.
    The latter condition is equivalent to the membership in $\cL_{\rho,i-1}$ and the acceptance by $\cA_{\rho,i-1}$.
    Then we immediately have (ii) and the formula for $\cP_{\rho,i}$ in (i).

    Again by definition, $\cS_{\rho,i}$ consists of all strings that have the $k$-suffix $u$ and, after deleting the last letter, avoid the set $\{\rho[1],\ldots,\rho[i]\}$.
    Thus $\cS_{\rho,i}$ consists of all strings $xva$, where $xv\in\cL_{\rho,i}$; this gives us the formula for (i).
    The number of such strings equals the number of the strings $xv$ accepted by $\cA_{\rho,i}$.
    Since $va=\rho[i]$ is in the antidictionary of $\cL_{\rho,i}$, $v$ is a vertex in $\cA_{\rho,i}$, and no vertices $v'$ with $|v'|>|v|$ exist.
    Then every walk in $\cA_{\rho,i}$, having the label with the suffix $v$, ends in the vertex $v$ by the definition of canonical DFA.
    Hence all accepting walks in $\cA_{\rho,i}$ labelled by the strings of the form $xv$ are $(\lambda,v)$-walks, which implies (iii).  
\end{proof}

By Lemma~\ref{l:regular}(i), $g(\cP_{\rho,i{+}1}), g(\cS_{\rho,i})\le g(\cL_{\rho,i})=\alpha_{\rho,i}$.
The following lemma shows that under a natural condition these inequalities turn into equalities.

\begin{lemma} \label{l:PSformulas}
    Suppose that $\rho\in\U_{\sigma,k}$, $i\ge 1$, and the DFA $\cA_{\rho,i}$ is a simple digraph with the second largest eigenvalue $\gamma$ and the limit matrix $\A$.
    Then \\
    \emph{(i)} if $\rho[i]=u=va$ for $a\in\Sigma$, then $C_{\cS_{\rho,i}}(n)=\A[\lambda,v]\cdot\alpha_{\rho,i}^{n-1}+\tilde{O}(|\gamma|^n)$; if, moreover, $\alpha_{\rho,i}<\alpha_{\rho,i{-}1}$ and $\cA_{\rho,i-1}$ is simple, then $\A[\lambda,v]>0$;\\
    \emph{(ii)} if $\rho[i+1]=z$, then $C_{\cP_{\rho,i{+}1}}(n)=\big(\sum_{x\in \cA_{\rho,i}} \A[\lambda.z,x]\big)\cdot\alpha_{\rho,i}^{n-k}+\tilde{O}(|\gamma|^n)$; if, moreover, $\alpha_{\rho,i+1}<\alpha_{\rho,i}$, then $\sum_{x\in \cA_{\rho,i}} \A[\lambda.z,x]>0$.
\end{lemma}
\begin{proof}
    Plugging Lemma~\ref{l:regular}(ii,iii) into Lemma~\ref{l:limit}, we immediately get the formulas for the growth functions in both (i) and (ii).
    Let us prove the inequalities for the elements of $\A$.

    The DFA $\cA_{\rho,i}$ contains $\Omega(\alpha_{\rho,i}^n)$ walks of length $n$ within the major component; all these walks are labelled by the strings in $\cL_{\rho,i}$.
    Since the language $\cL_{\rho,i+1}=\cL_{\rho,j}-\Sigma^*z\Sigma^*$ has the growth rate $\alpha_{\rho,i+1}<\alpha_{\rho,i}$, some walk in the major component of $\cA_{\rho,i}$ has the label $z$.
    Since all vertices of $\cA_{\rho,i}$ are strings shorter than $z$, such a walk ends in the longest suffix of $z$ that is a vertex, i.e., in the vertex $\lambda.z$.
    Hence $\lambda.z$ is in the major component of $\cA_{\rho,i}$, and thus is the starting point of $\Omega(\alpha_{\rho,i}^n)$ walks of length $n$; then $\sum_{x\in \cA_{\rho,i}} \A[\lambda.z,x]>0$ and (ii) is proved.

    We build an auxiliary DFA $\cA$ as follows: perform two first steps of constructing the canonical DFA $\cA_{\rho,i}$ and then delete all leaves \emph{except for $u$} in the third step ($u$ is labeled as terminal).
    Then $\cA$ accepts the language $\cL_{\rho,i-1}$ and can be transformed into $\cA_{\rho,i}$ by deleting the vertex $u$.
    Let us show that $\cA$ is simple.
    Suppose it is not.
    Then it has two components $H_1,H_2$ such that $\ind(H_1), \ind(H_2)>1$ and $H_2$ is not reachable from $H_1$.
    Let $x\in H_1, y\in H_2$.
    There are exponentially many walks in $\cA$ ending in $y$.
    Then the labels of exponentially many of them have some common suffix $\bar y$ of length $k+1$.
    Since $|y|\le k$, $y$ is a proper suffix of $\bar y$.
    Then all walks in $\cA$ with the labels of the form $s\bar y$ end in $y$; as a result, no such walk starts in $x$. 
    Since $x$ is a vertex of $\cA$, we have $\lambda.x=x$ and hence $\cA$ accepts no strings of the form $xs\bar y$.
    On the other hand, it accepts exponentially many strings of each of four types: with the prefix $x$, with the suffix $x$, with the prefix $\bar y$, and with the suffix $\bar y$.
    All these strings are also accepted by $\cA_{\rho,i-1}$, as it shares the accepted language $\cL_{\rho,i-1}$ with $\cA$.
    Since $\cA_{\rho,i-1}$ is simple, its vertices $\lambda.x$ and $\lambda.\bar y$ both belong to its major component.
    Then $\cA_{\rho,i-1}$ accepts arbitrarily long strings of the form $xs\bar y$.
    This contradicts the fact that $\cA$ and $\cA_{\rho,i-1}$ accept the same language.
    Hence our assumption was wrong and $\cA$ is indeed simple.
    
    The deletion of the vertex $u$ from $\cA$ reduces the index of the DFA: $\ind(\cA)=\alpha_{\rho,i-1}>\alpha_{\rho,i}=\ind(\cA_{\rho,i})$.
    Hence $u$ belongs to the major component $H$ of $\cA$.
    As the only ingoing edge of $u$ is $v\xrightarrow{a}u$, one has $v\in H$.
    Since in $\cA$ $v$ is reachable from all vertices of $H$ and $u$ is reachable only through $v$, in $\cA_{\rho,i}$ $v$ is reachable from all vertices of the set $H-\{u\}$. 
    All vertices in $\cA$ outside $H$ belong to trivial components, hence the vertices of the major component of $\cA_{\rho,i}$ belong to the set $H-\{u\}$.
    Hence $v$ is reachable from the major component of $\cA_{\rho,i}$, which is in turn reachable from the vertex $\lambda$.
    Therefore, the number of $(\lambda,v)$-walks of length $n$ in $\cA_{\rho,i}$ is $\Omega(\alpha_{\rho,i}^n)$, yielding $\A[\lambda,v]>0$ and the statement (i).
\end{proof}

Since $\cL_{\rho,i+1}\subseteq \cL_{\rho,i}$ for all $i\in[1..|\rho|{-}1]$, the sequence $\{\alpha_{\rho,i}\}$ is (non-strictly) decreasing.
Accordingly, we approximate the language $\cW_\rho$ of all charged windows with the sequence of languages $\cW_{\rho,i}=\bigcup_{j=1}^{i\pmb{+1}}\cP_{\rho,j}\cup \bigcup_{j=1}^{i}\cS_{\rho,j}$.
Then $g(\cW_\rho-\cW_{\rho,i})\le \alpha_{\rho,i+1}$ by Lemma~\ref{l:PSformulas}.
The following lemma is our main tool for the asymptotic comparison of arrangements.

\begin{lemma} \label{l:comparison}
    Let $\pi,\rho\in\U_{\sigma,k}$.
    If there exist $i,j\in\mathbb{N}$, $\varepsilon>0$ such that $C_{\cW_{\pi,j}}(n)-C_{\cW_{\rho,i}}(n)=\Omega((\alpha_{\rho,i+1}+\varepsilon)^n)$, then $\rho\triangleleft \pi$.
\end{lemma}
\begin{proof}
    Since $g(\cW_\rho-\cW_{\rho,i})\le \alpha_{\rho,i+1}$ by Lemma~\ref{l:PSformulas}, we have $C_{\cW_\pi}(n)-C_{\cW_\rho}(n)\ge C_{\cW_{\pi,j}}(n)-C_{\cW_\rho}(n)=\Omega((\alpha_{\rho,i+1}+\varepsilon)^n)$, implying $d_{(\rho,w)}< d_{(\pi,w)}$ for all $w$ big enough.
\end{proof}

The growth rate of a regular language is a zero of a polynomial, so in general it cannot be found exactly.
However, indices of digraphs, and hence growth rates of regular languages, can always be compared (see, e.g., \cite{Shur10imm}).
For details, see Appendix~\ref{ss:proofs}, Remark~\ref{r:compare_indices}.

\subsection{Eventually optimal orders for small $\sigma$ and $k$}
\label{ss:optimal}

For a UHS order $\rho\in\mathbb{U}_{\sigma,k}$, let $i\in[1..|\rho|]$ be such that $\alpha_{\rho,i-1}>\alpha_{\rho,i}=1$.
Then $\pi=\rho[1..i]$ is the \emph{head} of $\rho$; we write $\rho=\pi\cdot\pi'$ and refer to $\pi'$ as the \emph{tail} of $\rho$.
Then for sufficiently large $w$ the density of $\rho$ is determined by its head $\pi$, with a negligible additional effect of $\pi'$.
We call the head $\pi$ \emph{good} if $\alpha_{\pi,j-1}<\alpha_{\pi,j}$ for all $j\le i$.
We say that an eventually optimal UHS order $\rho\in\mathbb{U}_{\sigma,k}$ is \emph{essentially unique} if its head $\pi$ is a unique, up to renaming the letters, good head of an eventually optimal order of $\Sigma^k$.
In this section, we compute eventually optimal UHS orders and prove their essential uniqueness.
In the simplest case $\sigma=k=2$, one order is optimal for \emph{all} $w$ (not just eventually optimal).

\begin{restatable}{theorem}{asymptotictwo} \label{t:22w_opt}
    (1) The minimum density of a $(2,2,w)$-minimizer is $\frac{2^w+w+5}{2^{w+2}}$.\\
    (2) The order $\rho=(01,10,00,11)\in\U_{2,2}$ is optimal for all $w\ge2$ and essentially unique.
\end{restatable}

For all remaining $(\sigma,k)$ pairs, we build eventually optimal UHS orders step by step, comparing arrangements by the $\triangleleft$ relation to build consecutively the sets $\cand_1,\cand_2,\ldots$.
By the definition of $\triangleleft$, it is safe to choose the arrangements for $\cand_{i+1}$ only among those having prefixes in $\cand_i$.
Our main tool is Lemma~\ref{l:comparison}.
We also use two auxiliary rules to narrow the search; see Appendix~\ref{sss:search}.
The following lemma, which compiles several quite technical results by Guibas and Odlyzko \cite{GuOd78, GuOd81}, restricts the choice of $\cand_1$.

\begin{lemma}[\!\cite{GuOd78, GuOd81}] \label{l:GuOd}
Let $u,u'\in\Sigma^k$, $L_{u},L_{u'}\subset \Sigma^*$ be the languages avoiding $u$ and $u'$ respectively.
If $u$ and $u'$ have the same set of periods, then $C_{L_u}(n)=C_{L_{u'}}(n)$ for all $n$.
If $u$ has no periods and $u'$ has at least one period, then $g(L_{u})<g(L_{u'})$.    
\end{lemma}

\begin{restatable}{corollary}{corGO} \label{cor:cand1}
   Let $\pi,\rho\in\U_{\sigma,k}$ such that $\rho[1]$ has no periods.
   Then\\
   \emph{(i)} $\rho\triangleleft\pi$ if $\pi[1]$ has a period; \emph{(ii)} $C_{\cS_{\pi,1}}(n)=C_{\cS_{\rho,1}}(n)$ for all $n$ if $\pi[1]$ has no periods.  
\end{restatable}

By Corollary~\ref{cor:cand1}, $\cand_1=\{(u)\mid u\in\Sigma^k \text{ is lexmin and has no periods}\}$ is a valid set of 1-candidate arrangements of $k$-mers.
For $\sigma=2$, this set contains two ``exceptional'' strings $0^{k-1}1$ and $01^{k-1}$: their canonical DFAs are not strongly connected (Fig.~\ref{f:aut01}).
(As long as $\sigma,k\ge 2$, the canonical DFA avoiding any other lexmin $k$-mer is strongly connected.)
The following lemma about these strings is helpful for the main theorems of this section.

\begin{restatable}{lemma}{exceptional} \label{l:0001}
    Let $u\in\{0,1\}^k$, $a\in\{0,1\}$, $v=a01^{k-2}$.
    Then $(01^{k-1},v)\triangleleft(0^{k-1}1,u)$.
    In addition, if $u\ne a01^{k-2}$ for $a\in\{0,1\}$, then $(01^{k-1},v)\triangleleft(01^{k-1},u)$.
\end{restatable}

\begin{figure}[!htb]
    \centering
    \includegraphics[scale=0.82, trim = 45 730 70 41, clip]{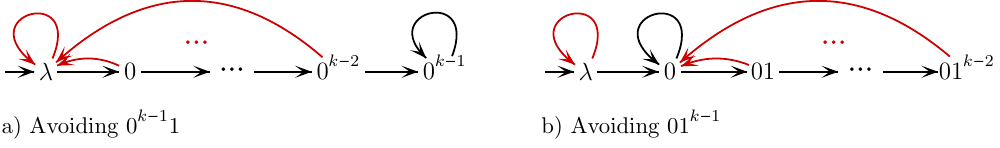}
    \caption{Canonical DFAs avoiding exceptional $k$-mers $0^{k-1}1$ and $01^{k-1}$ (Lemma~\ref{l:0001}). Black (red) edges are labelled by 0 (resp., by 1).}
    \label{f:aut01}
\end{figure}

Now we present optimal UHS orders for the announced $(\sigma,k)$ pairs.

\begin{restatable}{theorem}{asymptoticthree}  \label{t:23w_opt}
    The UHS order $\rho=(011,001,101,000,110,111)\in\U_{2,3}$ is eventually optimal and essentially unique.
\end{restatable}
\begin{proof}
By Corollary~\ref{cor:cand1}, we take $\cand_1=\{(001),(011)\}$.
For the list $\cand_2$, we consider the extensions of the elements of $\cand_1$.
By Lemma~\ref{l:0001}, it suffices to consider the arrangements $\pi=(011,001)$ and $\pi'=(011,101)$; see Fig.~\ref{f:aut_k3} for the automata $\cA_{\pi,2}$ and $\cA_{\pi',2}$.
We observe that $\cA_{\pi,2}$ is flat, while $\cA_{\pi',2}$ is simple, with the major component on the vertices $0,01,10$.
Then $\alpha_{\pi,2}=1<\alpha_{\pi',2}\approx 1.4656<\alpha_{\pi',1}\approx 1.6180$.
By Lemma~\ref{l:PSformulas}(i), $C_{\cS_{\pi',2}}(n)=\Omega(\alpha_{\pi',2}^n)$.
We apply Lemma~\ref{l:comparison} to $\pi',\pi$ with $i=1,j=2,\varepsilon=\alpha_{\pi',2}-1$, obtaining $\pi\triangleleft\pi'$.
Accordingly, we take $\cand_2=\{\pi\}$.
In particular, every eventually optimal lexmin UHS order has the head $\pi$, so we proved essential uniqueness of such an order.
It remains to study the tail, which is done in the full proof in Appendix~\ref{sss:main}.
\end{proof}
\begin{figure}[!htb]
    \centering
    \includegraphics[scale=0.82, trim = 45 700 220 48, clip]{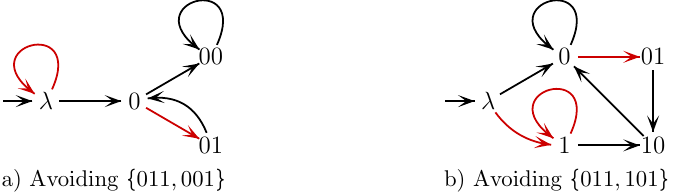}
    \caption{Canonical DFAs avoiding the sets $\{011,001\}$ and $\{011,101\}$ (Theorem~\ref{t:23w_opt}). Black (red) edges are labelled by 0 (resp., by 1).}
    \label{f:aut_k3}
\end{figure}

\begin{restatable}{theorem}{asymptotic} \label{t:24w_opt}
    The UHS order $\rho=(0011,0001,1100,0100,1110,1011,0000,0101,1111)\in\U_{2,4}$ is eventually optimal and essentially unique.
\end{restatable}

\begin{restatable}{theorem}{asymptoticfive} \label{t:25w_opt}
    The UHS order $\rho=(01011, 00101, 10101, 00010, 11010, 10010, 11001, 11100$, $11110, 00111, 10111, 00001, 01100, 11011, 00000, 11111)\in\U_{2,5}$ is eventually optimal and essentially unique.
\end{restatable}

\begin{restatable}{theorem}{asymptoticdna} \label{t:42w_opt}
    The UHS order $\rho=(01, 20, 30, 10, 21, 32, 12, 31, 13, 33, 22, 02, 00, 11)\in\U_{4,2}$ is eventually optimal and essentially unique.
\end{restatable}

\subsection{On density lower bounds}
\label{ss:lower}

Given $\rho\in\U_{\sigma,k}$, we can approximate its set of charged windows $\cW_\rho$ by the subsets $\cW_{\rho,i}$, $i=0,\ldots,|\rho|-1$. 
The set $\cW_{\rho,0}$ consists of $\sigma^w$ windows with the prefix $\rho[1]$, implying the trivial lower bound $d_{(\rho,w)}\ge \frac{1}{\sigma^k}$.
Note that $g(\cW_\rho-\cW_{\rho,0})=\alpha_{\rho,1}$ due to Lemma~\ref{l:PSformulas}.
Let $i$ be such that $\alpha_{\rho,1}=\cdots=\alpha_{\rho,i}>\alpha_{\rho,i+1}$.
We 
denote $B(\rho,n)=C_{\cW_{\rho,i}}(n)$.
Since $\frac{B(\rho,w+k)}{\sigma^{w+k}}\le \frac{C_{\cW_{\rho}}(w+k)}{\sigma^{w+k}}=d_{\rho,w}$, the function 
\begin{equation} \label{e:lb}
\beta_1(\sigma,k,w)=\min_{\rho\in\U_{\sigma,k}}\tfrac{B(\rho,w+k)}{\sigma^{w+k}}
\end{equation}
is a density lower bound.
For small $\sigma,k$, computing $\beta_1(\sigma,k,w)$ is feasible, so we compared this bound to the minimum density computed by Theorem~\ref{t:exhaustive} (see Fig.~\ref{f:plots}).
The plots show that $\beta_1(\sigma,k,w)$ is very close to the minimum density for big $w$.
Thus, obtaining a formula for $\beta_1(\sigma,k,w)$ seems an important task.
We partially solve it with the following theorem.

\begin{restatable}{theorem}{lbound} \label{t:lower}
    Let $\sigma, k\ge2$ be fixed, $\sigma+k>4$. 
    There exists $w_0$ such that for all $w\ge w_0$ one has $\beta_1(\sigma,k,w)=\frac{1}{\sigma^k}+\frac{c_{\sigma,k}\alpha_{\sigma,k}^{w+k}+O(1)}{\sigma^{w+k}}$,
    where $\alpha_{\sigma,k}$ is the largest root of the polynomial $r^k-\sigma r^{k-1}+1$ and $c_{\sigma,k}$ is a positive constant.
    If $\sigma>2$, $O(1)$ can be replaced by $o(1)$.
\end{restatable}
\begin{proof}
    (Sketch; see Appendix~\ref{ss:lboundproof} for the full proof.)
    We first observe that $\alpha_{\rho,2}<\alpha_{\rho,1}$ for all $\rho$ except for the case where $\rho[1]=0^{k-1}1$ and $\sigma=2$.
    But in this case the minimum in \eqref{e:lb} cannot be reached on $\rho$ because of Lemma~\ref{l:0001}, so we can assume  $B(\rho,n)=\sigma^w+C_{\cS_{\rho,1}}(n)+C_{\cP_{\rho,2}}(n)$.

    Let $\alpha=\min_{\rho\in\U_{\sigma,k}} \alpha_{\rho,1}$.
    Starting from some $w=w_0$, the minimum in \eqref{e:lb} is reached on some $\rho$ with $\alpha_{\rho,1}=\alpha$. 
    We use Lemma~\ref{l:GuOd} and take a particular $k$-mer $u$ to find the characteristic polynomial $\chi(r)=r^k-\sigma r^{k-1}+1$ of the DFA $\cA_u$, with the root $\alpha$.
    We then prove, by Rouch\'e's theorem, that all other zeroes of $\chi(r)$ have absolute values $<1$ (plus a simple zero $r=1$ in the case $\sigma=2$).
    From that, we derive that if the minimum in \eqref{e:lb} for some $w\ge w_0$ is reached on $\bar\rho$, then $\cA_{\bar\rho,1}$ has the characteristic polynomial $\chi(r)$ (or its variation $\bar\chi(r)=\frac{r-c}{r-1}\chi(r)$, where $c\in\{-1,0,1\}$, if $\sigma=2$).
    Then by Lemma~\ref{l:regular} both $C_{\cS_{\bar\rho,1}}(n)$ and $C_{\cP_{\bar\rho,2}}(n)$ satisfy the linear recurrence with the characteristic polynomial $\chi(r)$ or $\bar\chi(r)$.
    Due to the location of the roots of $\chi(r)$, the solutions of this recurrence look like the numerator in the statement of the theorem.
    Finally, Lemma~\ref{l:PSformulas}(i) guarantees that $c_{\sigma,k}>0$.
\end{proof}

\section{Results of Computer Search} \label{s:experiments}

We implemented the search algorithm of Theorem~\ref{t:exhaustive}, with all additional tricks described in Section~\ref{s:exhaust}, and ran it for small $(\sigma,k)$ pairs over big ranges of $w$.
(See Appendix~\ref{ss:scripts} for the scripts.)
Our goal was to get a complete picture of the minimum density and the evolution of minimum density minimizers for the $(\sigma,k)$ pairs that are feasible for our search algorithm.
One of the important points was to find, for each $(\sigma,k)$ pair, the smallest window size $w_\infty(\sigma,k)$ for which the eventually optimal order for $(\sigma,k)$ reaches the minimum density. 
The case $\sigma=k=2$ needs no search, as $w_\infty(2,2)=2$  by Theorem~\ref{t:22w_opt}.
The case $\sigma=2,k=3$ can be  brute-forced with no optimization, as the search space is small.
Here, one UHS order is optimal for all $w\in[2..8]$, and the UHS order from Theorem~\ref{t:23w_opt} is optimal since $w_\infty(2,3)=8$.
Respectively, we focused on the nontrivial $(\sigma,k)$ pairs: $(2,4), (2,5)$, and $(4,2)$.

1. In each case, we computed the minimum density of a minimizer for all $w\in[2..3\sigma^k]$ and compared it to the average density and to the lower bounds (Fig.~\ref{f:plots}).
The plots show \emph{density factors}, which are densities normalized by the window size: $df_{(\rho,w)}=(w+1)d_{(\rho,w)}$.
We point out three features:
\begin{itemize}
    \item plots look similar; note the same vertical range (in absolute numbers) and the same horizontal range (in multiples of $\sigma^k$) of all plots;
    \item the gap between minimal and average density becomes smaller on the increase of $k$ or $\sigma$;
    \item the lower bound \eqref{e:lb} is very good for big $w$.
\end{itemize}
\begin{figure}[!htb]
    \centering
    \includegraphics[scale=0.34]{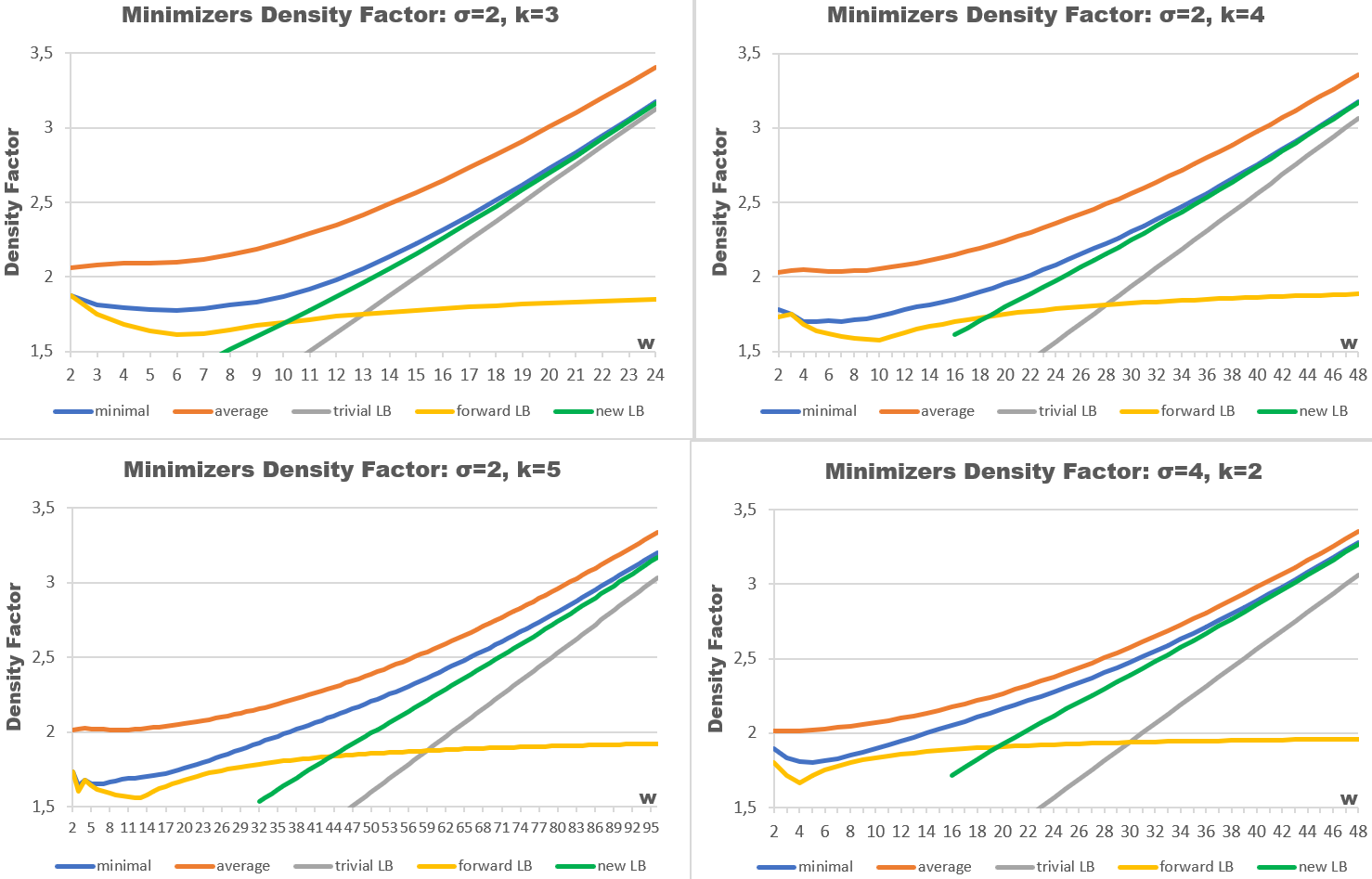}
    \caption{Density factors of minimizers over big ranges of $w$. The minimum (blue) and average (orange) densities are compared to lower bounds: trivial $\frac{1}{\sigma^k}$ (grey), forward schemes bound \cite{KKMLT24lower} (yellow), and the bound $\beta_1$ (see \eqref{e:lb}; green). All values are normalized by the factor of $w+1$.}
    \label{f:plots}
\end{figure}

2. The search described above revealed $w_\infty(4,2)=25$.
For the cases $\sigma=2,k=4,5$, we extended the computations to larger window sizes and found $w_\infty(2,4)=56$, $w_\infty(2,5)=262$.
From \cite[Prop.\,3]{GoSh25} we know that the \textbf{average} density is $\frac{1+o(1)}{\sigma^k}$, i.e., close to the trivial lower bound, whenever $w=\frac{\sigma}{\sigma-1}N(\ln N+h(N))$, where $N=\sigma^k$ and $h$ is any infinitely growing function.
We observe that $w_\infty(2,5)>2\cdot2^5\cdot \ln 2^5\approx 222$.
If this is not a unique case but a general rule, then the eventually optimal orders have purely theoretical interest.

3. The optimal orders found for all studied $(\sigma,k)$ pairs are listed with their densities in Appendix~\ref{ss:listoptimal}. 
The evolution of optimal UHS orders with the growth of $w$ can be summarized as follows.
The range $[2..w_\infty(4,2){-}1]$ is partitioned into 7 smaller ranges of sizes from 1 to 6 with different optimal UHS orders for $\sigma=4,k=2$; all these orders have different lexmin heads.
For $\sigma=2,k=4$, there are 5 such ranges: $[2], [3..6],[6..12],[13..26]$, and $[27..55]$.
Note that the range size grows as $w$ increases and two ranges share an endpoint. 
Within the range $[6..12]$, there is an oscillation with period 2: for the same head, one tail gives the minimal density for odd $w$, and a different one for even $w$.
Finally, for $\sigma=2,k=5$, we identified 18 ranges with different lexmin heads (if for some $w$ optimal orders with different lexmin heads exist, we chose the one decreasing the number of ranges).
The largest range is $[136..223]$.
Inside the ranges, several tail effects were spotted.
There are oscillations (periods from 2 to 6 over 2 to 4 different tails), a range split (two tails give minimum densities over two subranges of the range $[12..17]$), and an irregular behaviour: there is a weird example of a UHS order that is optimal for $w=3,4,9$, but not for intermediate values.
The overall conclusion is that studying heads of orders is important to understand the structure of optimal orders for different $(\sigma,k,w)$ tuples, while the tail effects are marginal.

4. Somewhat counterintuitive, the most resource-consuming search cases are those of small $w$.
For big $w$, the use of a good upper bound (see Section~\ref{s:exhaust}) drastically shrinks the set of processed subsets of $\Sigma^k$, saving both time and space.
In the resource-critical case $\sigma=2,k=5$, the window size $w=2$, and only it, appeared infeasible for a laptop with 8 Gb RAM.
However, in this case some orders constructed by the greedy algorithm of \cite{GTKOS25greedymini} are provably optimal as they hit the lower bound of \cite{KKMLT24lower}, so we just took one of them.

\section{Discussion an Future Work}

In this paper, we made a step to the systematic study of minimum density of minimizers and optimal orders reaching this density.
For the first time, several infinite-size cases are covered: for several $(\sigma,k)$ pairs we presented minimum densities and optimal orders for all $w$.
Due to a novel approach to minimizers through regular languages with finite antidictionaries, we were able to ``reach the infinity'' in the window size.

There is a plenty of natural questions for further study.
First, does there exist an algorithm computing the minimum density in single-exponential time? Can any hardness result be proved for this problem?
Second, what are the lower and upper bounds on the minimum window size $w_\infty(\sigma,k)$, where the optimal order is eventually optimal?
Third, which $(\sigma,k)$ pairs are feasible for the computer search of the eventually optimal order? Are there cases where the asymptotic optimality cannot be proved (i.e., Lemma~\ref{l:comparison} is insufficient to compare orders)?
Fourth, can the regular languages approach be used to efficiently obtain low-density minimizers for practically important $(k,w)$ pairs?
Fifth, can Theorem~\ref{t:lower} on the lower bound be strengthened by providing an explicit value $w_0$ and an asymptotic estimate to the constant $c_{\sigma,k}$?

\bibliography{bib}

\appendix

\section{Omitted Proofs and Details} \label{ss:proofs}

\subsection{To Section~\ref{s:exhaust}}

\densitylemma*
\begin{proof}
    
Two different approaches described below result in the algorithms counting the windows charged by $\pi$ due to the $k$-mer $u=\pi[i]$ in time $O(\sigma^w)$ and $O(w\sigma^k)$, respectively.
Note that 
\begin{itemize}
    \item the windows prefix-charged by $\pi$ due to $u$ are exactly those having the $k$-prefix $u$ and all other $k$-mers of rank at least $i$;
    \item  the windows suffix-charged by $\pi$ due to $u$ are exactly those having the $k$-suffix $u$ and all other $k$-mers of rank at least $i{+}1$.
\end{itemize}

\noindent\textbf{DFS Algorithm.}
Consider the trie $\cT$ of all $(w{+}k)$-windows with the $k$-prefix $u$: it consists of the path from the root to $u$ and a complete $\sigma$-ary tree of depth $w$, rooted at the vertex $u$. 
Each leaf of $\cT$ corresponds to a window; to count prefix-charged windows with the prefix $u$, we run a recursive DFS on the complete subtree of $\cT$.
Visiting a vertex $v$, we check the rank $r$ of its $k$-suffix, continuing the search if $r\ge i$ and skipping the subtree of $v$ if $r<i$ (as all leaves in this subtree are not prefix-charged).
If $v$ is a leaf and $r\ge i$, then $v$ is prefix-charged due to $u$, so we count it.
By the end of the search, we get the count of windows prefix-charged due to $u$.
For the suffix-charged windows, we consider the trie of reversals of all windows with the suffix $u$ and process it in a similar way.
The only difference is that the condition to skip the subtree becomes $r\le i$ instead of $r<i$.

Given the $k$-suffix of a vertex $v$, the $k$-suffixes of its children can be computed in $O(1)$ arithmetic operations (as $\sigma$ is a constant).
Then the total number of operations in each of two DFSs is proportional to the size of the tree, i.e., to $\sigma^w$.

\noindent\textbf{DP Algorithm}.
Recall that \emph{order-$k$ deBrujin graph over $\Sigma$} is a directed $\sigma$-regular graph, having all $\sigma$-ary $k$-mers as nodes and all pairs $(au,ub)$, where $u\in\Sigma^{k-1}$, $a,b\in\Sigma$, as edges.
If an edge $(au,ub)$ is labeled by $b$, the graph becomes a deterministic finite automaton $\cB$.
We view $\cB$ as a \emph{transition table} with rows indexed by $k$-mers and columns indexed by letters; the entry $\cB[u,a]$ contains the successor of $u$ by the letter $a$.

We complete the arrangement $\pi$ to a linear order $\rho$, assigning the remaining ranks lexicographically.
Then we replace all elements and all row indices in $\cB$ with their $\rho$-ranks, and sort the rows. 
The resulting table is referred to as $\cB_\rho$.

To count $(w+k)$-windows charged due to $u=\pi[i]$, we proceed by dynamic programming.
Let $\pref$ be a two-dimensional table such that $\pref[r,j]$ is the number of strings of length $k+j$ having the $k$-prefix $u$, the $k$-suffix of rank $r$, and all other $k$-mers of rank at least $i$.
Then the number of $(w+k)$-windows prefix-charged due to $u$ is $\sum_{r \ge i} \pref[r,w]$.
Similarly, let $\suff$ be a two-dimensional table such that $\suff[r,j]$ is the number of strings of length $k+j$ having the $k$-prefix of rank $r$, the $k$-suffix $u$ and all other $k$-mers of rank greater than $i$.
Then the number of $(w+k)$-windows suffix-charged due to $u$ is $\suff[i,w]$.

Note that $\pref[r,0]=[r=i]$, $\suff[r,0]=1$, and the DP rules for $\pref$ and $\suff$ are almost the same: $\pref[r,j+1]=\sum_\ell \pref[\ell,j]$, $\suff[r,j+1]=\sum_\ell \suff[\ell,j]$, where the summation for $\pref$ (resp., $\suff$) is over all ranks $\ell\ge i$ (resp., $\ell>i$) such that the $k$-mer of rank $r$ is a successor (resp., a predecessor) of the $k$-mer of rank $\ell$.
Both rules can be easily computed using the transition table $\cB_\rho$ to propagate the counts from the $j$'th column to the $(j+1)$th column along the edges of the deBrujin graph.

The computation requires $O(\sigma^k)$ time to build $\cB_\rho$ and $O(\sigma^k)$ time to compute one column of the $\pref$ and $\suff$ tables, to the total of $O(w\sigma^k)$.
\end{proof}

\average*

\begin{proof}
    By definitions, $\cR_\sigma(k,w)=\frac{1}{\sigma^{k+w}(\sigma^k)!}\sum_\rho \ch(\rho)=\frac{1}{\sigma^{k+w}(\sigma^k)!}\sum_\rho\sum_{t=1}^{\sigma^k} \ch(\rho,t)$, where $\rho$ runs over all orders on $\Sigma^k$.
    We swap the summation signs and compute the array $\sums[1..\sigma^k]$ such that $\sums[t]=\sum_\rho \ch(\rho,t)$, which is sufficient to obtain $\cR_\sigma(k,w)$.
    The algorithm computing $\sums$ uses the following property: by Lemma~\ref{l:locality}, $\ch(\rho,t)$ depends only on $\rho[t]$ and the domain of $\rho[1..t{-}1]$.
    It proceeds in phases and uses the algorithm from Lemma~\ref{l:ch} computing each number $\ch(\rho,i)$ in $O(w\sigma^k)$ time; this algorithm uses $O(\sigma^k)$ space.

    At the first phase, the algorithm computes $\ch((u),1)$ for each $u\in\Sigma^k$ and adds the results to $\sums[1]$.
    During the $t$'th phase, $t\ge 2$, the algorithm processes every $(t-1)$-element subset $U$ of $\Sigma^k$.
    It takes an arbitrary arrangement $\pi$ with the domain $U$, computes $\ch(\pi\cdot(u),t)$ for each $u\notin U$, and adds $\ch(\pi\cdot(u),t){\cdot}(t-1)!$ to $\sums[t]$.
    Since $\ch(\rho,t)=\ch(\pi,t)$ whenever $\rho[1..t{-}1]$ and $\pi[1..t{-}1]$ have the same domain and $\rho[t]=\pi[t]$, we have $\sums[t]=\sum_\rho \ch(\rho,t)$ after processing all $(t-1)$-element sets $U$. 

    The algorithm computes $2^{\sigma^k}$ values $\ch(\rho,t)$ and spends $O(\sigma^k)$ space both inside each of these computations and to store the array $\sums$.
    The theorem follows.
\end{proof}

\subsection{To Section~\ref{ss:growth}}

\limitmatrix*

\begin{proof}
    The matrix $\alpha^{-1}A$ has the strictly dominant eigenvalue 1, so its Jordan form can be written as 
    \arraycolsep=2pt
    $J=\left[\begin{array}{cccc}
         1& 0&\cdots& 0  \\
         \cline{2-4}
         \begin{matrix}
             0\\[-1.5mm] {\vdots} \\ 0
         \end{matrix}& 
         \multicolumn{3}{|c|}{J_0}\\
         \cline{2-4}
    \end{array}\right] $,
    where $J_0^n\xrightarrow[n\to\infty]{}0$.
    If $\alpha^{-1}A=TJT^{-1}$ for the transition matrix $T$, then the matrix $(\alpha^{-1}A)^n=TJ^nT^{-1}$ approaches $\A=\vec{x}\vec{y}^\top$, where $\vec{x}$ is the first column of $T$ and $\vec{y}^\top$ is the first row of $T^{-1}$.
    As $\A$ has rank 1, its columns are multiples of $\vec{x}$ and its rows are multiples of $\vec{y}^\top$.
    Since $(\alpha^{-1}A)\A=\A$, $\vec{x}$ is a column eigenvector of the matrix $\alpha^{-1}A$, belonging to the eigenvalue 1; hence, $\vec{x}$ belongs to the eigenvalue $\alpha$ of $A$.
    Similarly, the equality $\A(\alpha^{-1}A)=\A$ implies that $\vec{y}^\top$ is a row eigenvector of $A$, belonging to $\alpha$.
    Statement (i) is proved.

    The graph $G$ can be labeled as a DFA with the initial state $u$ and the only final state $v$.
    Hence, $W_{uv}(n)$ is the growth function of some regular language, and thus has the form \eqref{e:SaSo}.
    Since $W_{uv}(n)=A^n[u,v]$, we have $\frac{W_{uv}(n)}{\alpha^n}\xrightarrow[n\to\infty]{}\A[u,v]$, yielding statement (ii).
\end{proof}

\subsection{To Section~\ref{ss:regcharged}}
\begin{remark} \label{r:compare_indices}
    The growth rate of a regular language is a zero of a polynomial, so in general it cannot be found exactly, but can be approximated within any prescribed error range.
    Such approximate computation is sufficient to compare growth rates of two regular languages (or, more general, the indices of two digraphs) for $<,>,=$, using the following trick (see, e.g., \cite{Shur10imm}).
    Given characteristic polynomials $f(x)$ and $g(x)$ of two strongly connected digraphs, we represent them as $f(x)=f_1(x)h(x)$, $g(x)=g_1(x)h(x)$, where $h(x)$ is the greatest common divisor of $f(x)$ and $g(x)$.
    The maximum positive zeroes of $f(x)$ and $g(x)$ are simple, since the digraphs are strongly connected.
    If the maximum zero is common to $f(x)$ and $g(x)$, then the maximum zero of $h(x)$ is strictly greater than those of both $f_1(x)$ and $g_1(x)$.
    If the maximum zeroes of $f(x)$ and $g(x)$ are distinct, the maximum zeroes of $f_1(x)$ and $g_1(x)$ are also distinct, and at least one of them is strictly greater than the maximum zero of $h(x)$.
    Approximating the maximum zeroes of $f_1(x),g_1(x)$, and $h(x)$, we can distinguish between the above cases and thus compare the maximum zeroes of $f(x)$ and $g(x)$.
    If the graphs are not strongly connected, then we first split each of them into components and find the component with the largest index using the above procedure.
\end{remark}

\subsection{To Section~\ref{ss:optimal}}
\asymptotictwo*

\begin{proof}
    We first compute $d_{(\rho,w)}$.
    The rank-1 $k$-mer $01$ charges $2^w$ windows as a prefix and also $w+1$ windows $1^i0^{w-i}01$as a suffix, where $i\in[0..w]$.
    The remaining windows are of the form $1^i0^{w+2-i}$, where $i\in[0..w+2]$.
    Exactly four of them are charged: $10^{w+1}$ and $1^{w+1}0$ due to $10$, $0^{w+2}$ due to $00$ and $1^{w+2}$ due to $11$.
    Thus, $d_{(\rho,w)}=\frac{2^w+w+5}{2^{w+2}}$.
    Note that the same density can be achieved by assigning rank 2 to $00$: here $00$ charges the windows $1^w00$ and $0^{w+2}$, and both windows that avoid $\{10,00\}$ ($1^{w+1}0$ and $1^{w+2}$) are charged by any UHS order.
    Finally, if we assign rank 2 to $11$, then we charge all $(w{+}1)$ windows beginning with $11$, plus the window $0^{w+2}$, to the total of $w+2\ge4$.

    By symmetry, it remains to show that $d_{(\rho',w)}\ge d_{(\rho,w)}$ for every $\rho'\in\U_{2,2}$ with $\rho'[1]=00$.
    We denote the $n$'th Fibonacci number by $\Phi_n$ (with $\Phi_0=1$ and $\Phi_1=2$).
    It is well known that $\Phi_n$ is the number of binary $n$-strings avoiding $00$ (or $11$).
    The $k$-mer $00$ charges $2^w$ windows as a prefix and all windows $x100$, where $x\in\{0,1\}^{w-1}$ avoids $00$, as a suffix.
    The number of such strings $x$ is $\Phi_{w-1}$.
    The remaining windows are exactly those avoiding $00$.
    If $01$ has rank 2, it charges all windows of the form $01x$ and also the window $1^w01$; the remaining windows $1^{w+2}$ and $1^{w+1}0$ are both charged, to the total of $\Phi_w+3$ charged windows after rank 1.
    Similarly, we get $\Phi_w+2$ (resp., $\Phi_{w-1}+3$) charged windows if $11$ (resp., $10$) has rank 2; in the latter case, the window charged due to its prefix $10$ must have the form $101x$, so the number of such windows is $\Phi_{w-1}$.
    Finally, as $\Phi_w>\Phi_{w-1}\ge w$, we have at least $2^w+2\Phi_{w-1}+3\ge 2^w+2w+3\ge 2^w+w+5$ charged windows. 
    Hence $d_{(\rho',w)}\ge d_{(\rho,w)}$, and the theorem is proved.   
\end{proof}

\subsubsection{Auxiliary tools to build eventually optimal minimizers} \label{sss:search}

Let $\pi$ be an arrangement of $\Sigma^k$ with the domain $U$.
A set $K\subset\Sigma^k\setminus U$ is \emph{isolated} (w.r.t. $\pi$) if there is some $n_0\in\mathbb{N}$ such that every $n$-window with $n>n_0$, live w.r.t. $\pi$, contains either all or none of its $k$-mers in $K$;
the set of $n$-windows consisting of $k$-mers in the isolated set $K$ is also called isolated.
The following lemma is immediate.
\begin{lemma} \label{l:isolated}
    Let $\pi$ be an arrangement of $\Sigma^k$, $\rho_1=\pi\cdot\pi_1$ and $\rho_2=\pi\cdot\pi_2$ be two orders, and let $K$ be an isolated set for $\pi$, with the constant $n_0$.
    If, for every $u_1,u_2\in K$, the order of their $\rho_1$-ranks is the same as the order of their $\rho_2$-ranks, then $d_{\rho_1,w}=d_{\rho_2,w}$ for all $w>n_0-k$.
\end{lemma}
We call an isolated set $K$ \emph{small} if the number of windows having all $k$-mers in $K$ grows polynomially in $n$.
Let $\rho=\pi\cdot\pi'$ be an order and $K$ be a small isolated set w.r.t. $\pi$. A permutation $h$ of $\pi'$ is \emph{trivial} if it preserves the order of $k$-mers from $K$.
By Lemma~\ref{l:isolated}, $\pi\cdot h(\pi')$ has the same density as $\rho$ for all sufficiently big $w$.
This allows us not to choose the next $k$-mer from a small isolated set while there are other choices.
(As choosing such a $k$-mer affects only a small set of live windows, it does not bring us closer to finishing the head of a UHS order.)

Further, we call a $k$-mer $u$ the \emph{worst choice} for $\pi$, if it charges every window in $W$, in which it occurs.
Then an order with the prefix $\pi$ and \emph{not} maximal rank of $u$ cannot have smaller density than an order with the prefix $\pi$ and maximal rank for $u$.
Summarizing all the above, we formulate the following rule of reducing the search with no risk of losing optimality:
\begin{itemize}
    \item[$(\star)$] Never extend an arrangement $\pi$ to $\pi\cdot(u)$ if, w.r.t. $\pi$, $u$ is either (i) a worst choice, or (ii) an element of a small isolated set if other choices are available.
\end{itemize}

\corGO*

\begin{proof}
    Let $\pi[1]$ have a period.
    By Lemma~\ref{l:GuOd}, $\alpha_{\pi,1}=\alpha_{\rho,1}+\varepsilon$ for some $\varepsilon>0$.
    By Lemma~\ref{l:PSformulas}(i), $C_{\cS_{\pi,1}}(n)=\Omega(\alpha_{\pi,1}^n)$. 
    Then $C_{\cW_{\pi,1}}(n)-C_{\cW_{\rho,0}}(n)\ge C_{\cS_{\pi,1}}(n)=\Omega((\alpha_{\rho,1}+\varepsilon)^n)$.
    Taking $i=0,j=1$ in Lemma~\ref{l:comparison}, we get $\rho\triangleleft\pi$, proving (i).

    Now let $\pi[1]=u$ have no periods.
    Then all occurrences of $u$ in a string are disjoint.
    Hence if a string $v$ avoids $u$, then $u$ occurs in $vu$ only as a suffix.
    The map $v\to vu$ is a bijection of $\cL_{\pi,1}$ onto $\cS_{\pi,1}$, yielding $C_{\cS_{\pi,1}}(n)=C_{\cL_{\pi,1}}(n-k)$ for all $n\ge k$.
    Similarly, $C_{\cS_{\rho,1}}(n)=C_{\cL_{\rho,1}}(n-k)$ for all $n\ge k$.
    As $C_{\cL_{\rho,1}}(n-k)=C_{\cL_{\pi,1}}(n-k)$ by Lemma~\ref{l:GuOd}, we get statement (ii).
\end{proof}

\exceptional*

\begin{proof}
    Let $\pi=(0^{k-1}1,u)$, $\rho=(01^{k-1},v)$, $\alpha=\alpha_{\pi,1}=\alpha_{\rho,1}$, and let $A_0,\A_0$ denote the adjacency matrix and the limit matrix of the DFA $\cA_{\pi,1}$ (the vertices are ordered by length, i.e., left to right in Fig.~\ref{f:aut01}).
    Similarly, we write $A_1,\A_1$ for such matrices of $\cA_{\rho,1}$.
    We have
    $A_0=\begin{bmatrix}
        1&1&0&\cdots&0&0\\
        1&0&1&\cdots&0&0\\
        \vdots&\vdots&\vdots&\ddots&\vdots&\vdots\\
        1&0&0&\cdots&0&1\\
        0&0&0&\cdots&0&1\\
    \end{bmatrix}$,
    $A_1=\begin{bmatrix}
        1&1&0&\cdots&0&0\\
        0&1&1&\cdots&0&0\\
        \vdots&\vdots&\vdots&\ddots&\vdots&\vdots\\
        0&1&0&\cdots&0&1\\
        0&1&0&\cdots&0&0\\
    \end{bmatrix}$,
    and compute $\chi_{A_0}(r)=|rI-A_0|$, expanding along the first column of $rI-A_0$: $\chi_{A_0}(r)=(r^{k-1}-r^{k-2}-\cdots-r -1)(r-1)$.
    Hence $\alpha$ is the largest real number satisfying $\alpha^{k-1}=\alpha^{k-2}+\cdots+\alpha+1$; in particular, $\alpha=\sum_{i=0}^{k-2}\alpha^{-i}$.
    Then we compute  the column eigenvector $\vec{x}_0$ and the  row eigenvector $\vec{y}_0^\top$ of $A_0$, belonging to $\alpha$: $\vec{x}_0^\top=(\alpha,1+\frac{1}{\alpha}+\cdots \frac{1}{\alpha^{k-3}},\ldots,1+\frac{1}{\alpha},1,0)$, $\vec{y}_0^\top=(\alpha^{k-2},\alpha^{k-3},\ldots,\alpha,1,\frac{1}{\alpha-1})$. 
    By Lemma~\ref{l:limit}(i), $\A_0=\mu \vec{x}_0\vec{y}_0^\top$ for some constant $\mu>0$.
    Similarly, we get the same characteristic polynomial for $A_1$, compute the eigenvectors $\vec{x}_1^\top=(\frac{\alpha}{\alpha-1},\alpha,1+\frac{1}{\alpha}+\cdots \frac{1}{\alpha^{k-3}},\ldots,1+\frac{1}{\alpha},1)$, $\vec{y}_1^\top=(0,\alpha^{k-2},\alpha^{k-3},\ldots,\alpha,1)$, and write $\A_1=\mu' \vec{x}_1\vec{y}_1^\top$. 
    Let us show $\mu'=\mu$.

    Since $\cA_{\pi,1},\cA_{\rho,1}$ are simple DFAs, combining Corollary~\ref{cor:cand1}(ii) with Lemma~\ref{l:PSformulas}(i), we get $\A_0[\lambda,0^{k-1}]=\A_1[\lambda,01^{k-2}]$ (last elements of the first rows are equal).
    Note that the last elements of the first rows of the matrices $\vec{x}_0\vec{y}_0^\top$ and $\vec{x}_1\vec{y}_1^\top$ are also equal (to $\frac{\alpha}{\alpha-1}$).
    As $\mu^2(\vec{x}_0\vec{y}_0^\top)^2=\A_0^2=\A_0$, $\mu'^2(\vec{x}_1\vec{y}_1^\top)^2=\A_1^2=\A_1$, we have $\mu'^2=\mu^2$, and hence $\mu'=\mu$.

    In $\cA_{\rho,1}$, one has $\lambda.v=01^{k-2}$ (the last vertex in Fig.~\ref{f:aut01}b).
    Applying Lemma~\ref{l:PSformulas}(ii), we get $C_{\cP_{\rho,2}}(n)=\frac{\mu\vec{x}_1[k]}{\alpha^k}||\vec{y}_1||_1\cdot\alpha^n+\tilde{O}(\gamma^n)$, where $||\cdot||_1$ is the $l_1$-norm (sum of coordinates) of a vector and $|\gamma|<\alpha$. 
    We also note that $\alpha_{\rho,2}<\alpha$ since the major component of $\cA_{\rho,1}$ contains walks labelled by $\alpha$.
    Now let $\rho'=(01^{k-1},u)$, where $u\ne a01^{k-2}$.
    Then the vertex $\lambda.u$ in $\cA_{\rho',1}=\cA_{\rho,1}$ differs from $01^{k-2}$.
    As the last coordinate of $\vec{x}_1$ is the unique smallest one, $C_{\cP_{\rho',2}}(n)$ has a larger coefficient in $\alpha^n$ than $C_{\cP_{\rho,2}}(n)$.
    Since $\cS_{\rho',1}=\cS_{\rho,1}$, we apply Lemma~\ref{l:comparison} to $\rho'$ and $\rho$ with $i=j=1$ and $\varepsilon=\alpha-\alpha_{\rho,2}$ to conclude $\rho\triangleleft\rho'$.

    Similar to the above, we compute $C_{\cP_{\pi,2}}(n)$.
    First assume $\lambda.u\ne 0^{k-1}$ in $\cA_{\pi,1}$ (Fig.~\ref{f:aut01}a).
    As the smallest nonzero coordinate of $\vec{x}_0$ is $\vec{x}_0[k-1]$, we get $C_{\cP_{\pi,2}}(n)\ge \frac{\mu\vec{x}_0[k-1]}{\alpha^k}||\vec{y}_0||_1\cdot\alpha^n+\tilde{O}(\gamma'^n)$, where $|\gamma'|<\alpha$.
    We have $\vec{x}_0[k-1]\cdot||\vec{y}_0||_1=||\vec{y}_0||_1>\||\vec{y}_1||_1=\vec{x}_1[k]\cdot||\vec{y}_1||_1$.
    Since $C_{\cS_{\pi,1}}(n)=C_{\cS_{\rho,1}}(n)$ by Corollary~\ref{cor:cand1}(ii), we apply Lemma~\ref{l:comparison} to $\pi$ and $\rho$ with $i=j=1$ and $\varepsilon=\alpha-\alpha_{\rho,2}$, obtaining $\rho\triangleleft\pi$.

    Now assume $\lambda.u= 0^{k-1}$, i.e., $u\in\{0^k,10^{k-1}\}$. 
    By Lemma~\ref{l:PSformulas}(ii), $C_{\cP_{\pi,2}}(n)$ has zero coefficient in $\alpha^n$, and thus $\alpha_{\pi,2}=\alpha$.
    Consider the strings from $\cS_{\pi,2}$.
    Let $u=10^{k-1}$.
    In $\cA_{\pi,1}$ (Fig.~\ref{f:aut01},a), the $n$-strings having $u$ as a unique suffix are accepted by walks consisting of a $(\lambda,0^{k-2})$-walk of length $n-1$ followed by the edge $(0^{k-2},0^{k-1})$.
    The number of such walks is $\A_0[\lambda,0^{k-2}]\cdot\alpha^{n-1}+\tilde{O}(\gamma^n)$ by Lemma~\ref{l:limit}(ii), where $|\gamma|<\alpha$.
    Then the coefficient in $\alpha^n$ in the formula for $C_{\cS_{\pi,2}}(n)$ is $\frac{\mu}{\alpha}\vec{x}_0[1]\cdot\vec{y}[k-1]=\mu$.
    Since $\vec{x}_1[k]=1$ and $||\vec{y}_1||_1=\alpha^{k-1}$, the function $C_{\cP_{\rho,2}}(n)$ has the coefficient $\frac{\mu}{\alpha}<\mu$ in $\alpha^n$.
    Since $C_{\cS_{\pi,1}}(n)=C_{\cS_{\rho,1}}(n)$ by Corollary~\ref{cor:cand1}(ii), we apply Lemma~\ref{l:comparison} to $\pi$ and $\rho$ with $i=1$, $j=2$ and $\varepsilon=\alpha-\alpha_{\rho,2}$ to get $\rho\triangleleft\pi$.
    
    Finally, let $u=0^k$.
    In $\cA_{\pi,1}$, the $n$-strings having $u$ as a unique suffix are accepted by walks consisting of a $(\lambda,0^{k-2})$-walk of length $n-2$ followed by the edge $(0^{k-2},0^{k-1})$ and the loop at $0^{k-1}$.
    The number of such walks is asymptotically $\alpha$ times less than in the case $u=10^{k-1}$, resulting in the coefficient $\frac{\mu}{\alpha}$ in $\alpha^n$ in the formula for $C_{\cS_{\pi,2}}(n)$, which is the same as we have for $C_{\cP_{\rho,2}}(n)$.
    However, for every UHS order $\pi'$ with the prefix $\pi$ there exists $j>2$ such that $\alpha_{\pi',j}<\alpha_{\pi',j-1}=\alpha$.
    By Lemma~\ref{l:PSformulas}(ii), $C_{\cP_{\pi',j}}(n)$ has a positive coefficient in $\alpha^n$. 
    Therefore, we can apply Lemma~\ref{l:comparison} to $\pi'$ and $\rho$ with $i=1$, the obtained $j$, and $\varepsilon=\alpha-\alpha_{\rho,2}$ to conclude $\rho\triangleleft\pi'$.
    Since $\pi'$ is arbitrary, we get $\rho\triangleleft\pi$ by definition.
\end{proof}

\subsubsection{Main Theorems} \label{sss:main}

\asymptoticthree*

\begin{proof}
By Corollary~\ref{cor:cand1}, we take $\cand_1=\{(001),(011)\}$.
For the list $\cand_2$, we consider the extensions of the elements of $\cand_1$.
By Lemma~\ref{l:0001}, it suffices to consider the arrangements $\pi=(011,001)$ and $\pi'=(011,101)$; see Fig.~\ref{f:aut_k3} for the automata $\cA_{\pi,2}$ and $\cA_{\pi',2}$.
We observe that $\cA_{\pi,2}$ is flat, while $\cA_{\pi',2}$ is simple, with the major component on the vertices $0,01,10$.
Then $\alpha_{\pi,2}=1<\alpha_{\pi',2}\approx 1.4656<\alpha_{\pi',1}\approx 1.6180$.
By Lemma~\ref{l:PSformulas}(i), $C_{\cS_{\pi',2}}(n)=\Omega(\alpha_{\pi',2}^n)$.
We apply Lemma~\ref{l:comparison} to $\pi',\pi$ with $i=1,j=2,\varepsilon=\alpha_{\pi',2}-1$, obtaining $\pi\triangleleft\pi'$.
Accordingly, we take $\cand_2=\{\pi\}$.
In particular, every eventually optimal lexmin UHS order has the head $\pi$, so we proved essential uniqueness of such an order.
It remains to study the tail.

The strings from $\cL_{\pi,2}$ have the form $1^i(01)^j0^\ell$ for some $i,j,\ell\ge0$ (see Fig.~\ref{f:aut_k3}a).
The $3$-mer $101=\rho[3]$ charges $\lfloor\frac{n+1}{2}\rfloor$ windows of length $n$: the windows $1(01)^j0^\ell$, $j=1,\ldots,\lfloor\frac{n-1}{2}\rfloor$ as a prefix and $1^{n-2}01$ as a suffix.
The remaining $n$-windows are $010^{n-2}$ and $1^i0^{n-i}$, $i=0,\ldots,n$.
Among them, $0^n$, $1^{n-1}0$, and $1^n$ are charged by any order.
Moreover, it is easy to see that any choice of the rank-4 $3$-mer charges at least one remaining window except those three.
Every remaining window contains $\rho[4]$, $\rho[5]$, or $\rho[6]$, and exactly four remaining windows are charged by $\rho$: $0^n$ and $1^{n-3}0^3$ due to $000$, then $1^{n-1}0$ due to $110$, and finally $1^n$ due to $111$. 
In total, $\lfloor\frac{n+1}{2}\rfloor+4$ windows are charged by $\rho$ due to its tail $3$-mers, and this is the minimum possible for the fixed choice of $\rho[3]=101$.

Consider the alternative choices the $3$-mer of rank 3.
Choosing $111$ leads to charging $\Omega(n^2)$ windows with this prefix.
Assigning rank 3 to any of remaining $3$-mers: $000,010,100$, or $110$ charges either $\lfloor\frac{n+1}{2}\rfloor$ or $\lfloor\frac{n+2}{2}\rfloor$ windows due to this $3$-mer, and at least four remaining windows should be charged by any arrangement of remaining $3$-mers.
All cases are very similar, so we consider just one: let $000$ have rank 3.
Then $000$ charges the windows $1^i(01)^j0^3$, $j=0,\ldots,\lfloor\frac{n-3}{2}\rfloor$, plus the window $0^n$, to the total of $\lfloor\frac{n+1}{2}\rfloor$.
The remaining windows have the form $1^i(01)^j00, 1^i(01)^j0$, or $1^i(01)^j$, and among them $1^n$, $1^{n-1}0$, and one of the strings $1010\cdots$ and $0101\cdots$ are always charged.
In addition, if rank 4 is assigned to $010$ (resp., $101,110$), then $1^{n-3}010$ (resp., $1^{n-2}01$, $11010\cdots$) is charged, while assigning rank 4 to $100$ or $111$ results in $\Omega(n)$ additional charged windows.
Thus, the total number of windows charged due to tail $3$-mers is at least $\lfloor\frac{n+1}{2}\rfloor+4$, which is the result achieved by $\rho$.

We conclude that no UHS order with the prefix $\pi$ charges less $k$-mers than $\rho$.
Therefore, $\rho$ is indeed eventually optimal.
\end{proof}

\asymptotic*

\begin{proof}
    This proof uses the same tools as the proof of Theorem~\ref{t:23w_opt}, so we keep it concise.
    We proceed in iterations, building $\cand_i$ at $i$'th iteration.
    For $\pi\in\cand_i$, we write $\alpha_i=g(\cL_{\pi,i})$  (it is always the same number for all elements of $\cand_i$).\\
    \emph{Iteration 1}. We use Corollary~\ref{cor:cand1} to set $\cand_1=\{0001,0011,0111\}$.\\
    \emph{Iteration 2}. We compare the coefficients in $\alpha_1^n$ over all functions $C_{\cP_{\pi,2}}(n)$ (Lemma~\ref{l:PSformulas}(ii)), where $\pi$ has a prefix in $\cand_1$.
    The arrangements giving zero coefficient are excluded by Lemma~\ref{l:0001}; among the rest, $\pi_1=(0111, 0011)$, $\pi_2=(0111,1011)$, $\pi_3=(0011,0001)$, and $\pi_4=(0011,1001)$ give the minimal coefficient, so we exclude all other arrangements using Lemma~\ref{l:comparison}.
    Next, we compute $\alpha_2=\alpha_{\pi_1,2}=\alpha_{\pi_3,2}\approx 1.7221$, $\alpha_{\pi_2,2}\approx 1.7549$ and $\alpha_{\pi_4,2}\approx 1.8393$, thus excluding $\pi_2$ and $\pi_4$ using Lemma~\ref{l:comparison}.
    Finally we compute the coefficients in $\alpha_2^n$ for $C_{\cS_{\pi_1,2}}(n)$ and $C_{\cS_{\pi_3,2}}(n)$ (Lemma~\ref{l:PSformulas}(i)); they are equal. 
    Respectively, we set $\cand_2=\{\pi_1,\pi_3\}$ as we cannot prefer one to the other.\\
    \emph{Iteration 3}. 
    We proceed by the same scheme.
    The minimal nonzero coefficient in $\alpha_2^n$ among the functions $C_{\cP_{\pi,3}}(n)$, where $\pi$ has a prefix in $\cand_2$, is reached by $\pi_5=\pi_3\cdot(0100)$, $\pi_6=\pi_3\cdot(1100)$.
    We exclude every arrangement $\pi$ giving zero coefficient, using the same argument as in Lemma~\ref{l:0001}: $C_{\cS_{\pi,3}}(n)$ has a large coefficient in $\alpha_2^n$.
    Next, we compute $\alpha_3=\alpha_{\pi_6,3}\approx 1.6663<\alpha_{\pi_5,2}\approx 1.6736$, exclude $\pi_5$ by Lemma~\ref{l:comparison}, and set $\cand_3=\{\pi_6\}$.\\
    \emph{Iteration 4}. 
    The minimal nonzero coefficient in $\alpha_3^n$ among the functions $C_{\cP_{\pi,4}}(n)$, where $\pi$ has the prefix $\pi_6$, is reached solely by $\pi_7=\pi_6\cdot(0100)$; the arrangements giving zero coefficient, are excluded by the same argument as above.
    We compute $\alpha_4=\alpha_{\pi_7,4}\approx 1.6180$ and set $\cand_4=\{\pi_7\}$.\\
    \emph{Iteration 5}. 
    We first employ the rule $(\star)$, excluding from consideration the isolated group $\{1000,0000\}$ and the worst choice $1001$ (see Fig.~\ref{f:aut_k4},a).
    Among the other arrangements $\pi$ with the prefix $\pi_7$, the minimal coefficient in $\alpha_4^n$ in the function $C_{\cP_{\pi,5}}(n)$ is reached by $\pi_8=\pi_7\cdot(0010)$, $\pi_9=\pi_7\cdot(1010)$, $\pi_{10}=\pi_7\cdot(0110)$, and $\pi_{11}=\pi_7\cdot(1110)$.
    As $\alpha_5=\alpha_{\pi_{11},5}\approx1.3247$, $\alpha_{\pi_8,5}\approx1.6180$, $\alpha_{\pi_9,5}\approx1.4656$, $\alpha_{\pi_{10},5}\approx1.5129$, we set $\cand_5=\{\pi_{11}\}$.\\
    \emph{Iteration 6}. 
    We have the same isolated group and worst choice as at Iteration~5.
    The minimal nonzero coefficient in $\alpha_5^n$ is reached by the arrangement $\pi_{12}=\pi_{11}\cdot(1011)$.
    If $\pi\in\{\pi_{11}\cdot(0111),\pi_{11}\cdot(1111)\}$, it gives zero coefficient in $\alpha_5^n$ in $C_{\cP_{\pi,6}}(n)$, but a large coefficient in $\alpha_5^n$ in $C_{\cS_{\pi,6}}(n)$, so we exclude $\pi$.
    Hence we set $\cand_6=\{\pi_{12}\}$ and compute $\alpha_6=1$. 
    Thus, $\pi_{12}$ is the only lexmin candidate to the good head of an optimal order.\\
    \emph{Iteration 7}. 
    There are just 10 live windows in three isolated groups (Fig.~\ref{f:aut_k4},b): $\{10^{n-1},0^n\}$, $\{01^{n-1},1^n\}$, $\{(10)^{\frac{n}{2}},1(10)^{\frac{n-1}{2}},01(10)^{\frac{n-2}{2}},(01)^{\frac{n}{2}},0(01)^{\frac{n-1}{2}},10(01)^{\frac{n-2}{2}}\}$.
    Choosing the tail $\{0000,0101,1111\}$ (in any order), we charge just one window from a group, which is the minimum possible, as $0^n,1^n$, and one of $(10)^{\frac{n}{2}},(01)^{\frac{n}{2}}$ must be charged.
    Thus we finally obtained $\rho$ as an eventually optimal UHS order; uniqueness is proved at Iteration~6.
\end{proof}
\begin{figure}[!htb]
    \centering
    \includegraphics[scale=0.82, trim = 45 642 160 32.5, clip]{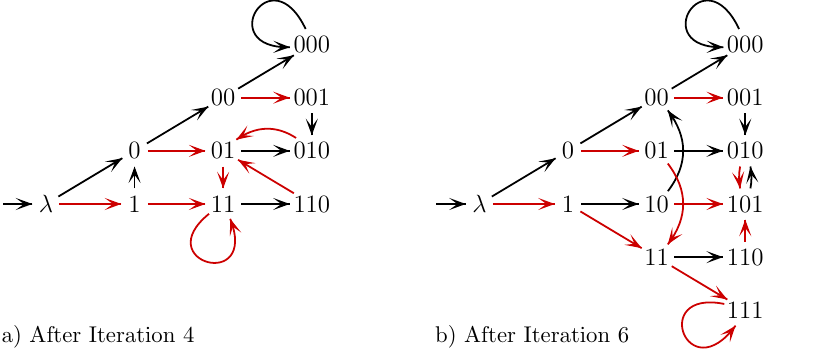}
    \caption{Canonical DFAs avoiding the current optimal arrangements $\pi_7$ and $\pi_{12}$ (Theorem~\ref{t:24w_opt}). Black (red) edges are labelled by 0 (resp., by 1).}
    \label{f:aut_k4}
\end{figure}

\asymptoticfive*

\begin{proof}
    We proceed in iterations, using the same arguments as in Theorem~\ref{t:24w_opt}. \\
    \emph{Iteration 1}. We use Corollary~\ref{cor:cand1} to set $\cand_1=\{00001,00011,00101,00111,01011,01111\}$.\\
    \emph{Iteration 2}. We compare the coefficients in $\alpha_1^n$ over the functions $C_{\cP_{\pi,2}}(n)$, where $\pi$ has a prefix in $\cand_1$, excluding by Lemma~\ref{l:0001} the arrangements giving zero coefficient. 
    Among the rest, the minimum is reached by
    $\pi_1=(01011, 00101)$ and $\pi_2=(01011,10101)$.
    We have $\alpha_2=\alpha_{\pi_1,2}\approx 1.8977<\alpha_{\pi_2,2}\approx 1.9380$ and respectively set $\cand_2=\{\pi_1\}$.\\
    \emph{Iteration 3}. 
    Among the coefficients in $\alpha_2^n$ over the functions $C_{\cP_{\pi,2}}(n)$ with $\pi$ having the prefix $\pi_1$, there are no zeroes and the minimum is reached solely at $\pi_3=\pi_1\cdot(10101)$. 
    Thus we set $\cand_3=\{\pi_3\}$ and compute $\alpha_3\approx1.8832$.\\
    \emph{Iteration 4}. 
    We observe that $01010$ is a worst choice: it appears in live windows only as a prefix.
    Among the coefficients in $\alpha_3^n$ over the functions $C_{\cP_{\pi,4}}(n)$ with $\pi$ having the prefix $\pi_3$, there are no zeroes and the minimum is reached at $\pi_4=\pi_3\cdot(00010)$, $\pi_5=\pi_3\cdot(10010)$, and $\pi_6=\pi_3\cdot(11010)$.
    We have $\alpha_4=\alpha_{\pi_4,4}=\alpha_{\pi_6,4}\approx1.8393<\alpha_{\pi_5,4}\approx1.8460$.
    Next we compute the coefficients in $\alpha_4^n$ in $C_{\cS_{\pi_4,4}}(n)$ and $C_{\cS_{\pi_6,4}}(n)$; the coefficient for $\pi_4$ is smaller, but we cannot drop $\pi_6$ yet, as there will be more coefficients in $\alpha_4^n$ at the next iteration.
    So we set $\cand_4=\{\pi_4,\pi_6\}$.\\
    \emph{Iteration 5}. 
    The coefficient in $\alpha_4^n$ over the functions $C_{\cP_{\pi,5}}(n)$ with $\pi$ having the prefix in $\cand_4$ reaches its minimum at $\pi_7=\pi_4\cdot(10010)$ and $\pi_8=\pi_4\cdot(11010)$.
    As we earlier established that $C_{\cS_{\pi_4,4}}(n)$ is asymptotically smaller than $C_{\cS_{\pi_6,4}}(n)$, we drop all arrangements with the prefix $\pi_6$.
    We compute $\alpha_5=\alpha_{\pi_8,5}\approx1.7902<\alpha_{\pi_7,5}\approx1.8051$ and respectively set $\cand_5=\{\pi_8\}$.\\
    \emph{Iteration 6}. 
    Among the coefficients in $\alpha_5^n$ over the functions $C_{\cP_{\pi,5}}(n)$ with $\pi$ having the prefix $\pi_8$, there are no zeroes and the minimum is reached solely at $\pi_9=\pi_8\cdot(10010)$. 
    Thus we set $\cand_6=\{\pi_9\}$ and compute $\alpha_6\approx1.7539$.\\
    \emph{Iteration 7}. 
    The minimum coefficient in $\alpha_6^n$ over the functions $C_{\cP_{\pi,6}}(n)$ where $\pi$ has the prefix $\pi_9$ is reached by arrangements $\pi_{10}-\pi_{15}$, obtained from $\pi_9$ by appending $00001,01001,01101,10001,11001$, and $11101$, respectively.
    We get $\alpha_7=\alpha_{\pi_{14},7}=\alpha_{\pi_{15},7}\approx1.6736$, with the other growth rates being larger. 
    In addition, $C_{\cS_{\pi_{14},7}}(n)=C_{\cS_{\pi_{15},7}}(n)$. 
    Hence we set $\cand_7=\{\pi_{14},\pi_{15}\}$.\\
    \emph{Iteration 8}. 
    The coefficient in $\alpha_7^n$ over the functions $C_{\cP_{\pi,8}}(n)$ with $\pi$ having the prefix in $\cand_7$ reaches its minimum at $\pi_{16}=\pi_{14}\cdot(01100)$, $\pi_{17}=\pi_{14}\cdot(11100)$, $\pi_{18}=\pi_{15}\cdot(01110)$ and $\pi_{19}=\pi_{15}\cdot(11110)$.
    We get $\alpha_8=\alpha_{\pi_{17},8}=\alpha_{\pi_{19},8}\approx1.5701<\alpha_{\pi_{16},8}=\alpha_{\pi{18},8}\approx1.6180$.
    The coefficients in $\alpha_8^n$ in $C_{\cS_{\pi_{17},7}}(n)$ and $C_{\cS_{\pi_{19},7}}(n)$ are equal, so we set $\cand_8=\{\pi_{17},\pi_{19}\}$.\\
    \emph{Iteration 9}. 
    The minimal nonzero coefficient in $\alpha_8^n$ is reached by the arrangements $\pi_{20}=\pi_{17}\cdot(01110)$ and $\pi_{21}=\pi_{17}\cdot(11110)$.
    If $\pi\in\{\pi_{19}\cdot(01111),\pi_{19}\cdot(11111)\}$, it gives zero coefficient in $\alpha_8^n$ in $C_{\cP_{\pi,9}}(n)$, but a large coefficient in $\alpha_8^n$ in $C_{\cS_{\pi,9}}(n)$, so we exclude $\pi$.
    We get $\alpha_9=\alpha_{\pi_{21},9}\approx1.4656<\alpha_{\pi_{20},9}\approx1.5214$ and respectively set $\cand_9=\{\pi_{21}\}$.\\
    \emph{Iteration 10}. 
    The minimal nonzero coefficient in $\alpha_9^n$ is reached by the arrangements $\pi_{22}=\pi_{21}\cdot(00111)$ and $\pi_{23}=\pi_{21}\cdot(10111)$.
    If $\pi\in\{\pi_{21}\cdot(01111),\pi_{21}\cdot(11111)\}$, it gives zero coefficient in $\alpha_9^n$ in $C_{\cP_{\pi,9}}(n)$, but a large coefficient in $\alpha_9^n$ in $C_{\cS_{\pi,9}}(n)$, so we exclude $\pi$.
    We get $\alpha_{10}=\alpha_{\pi_{22},10}\approx1.4313<\alpha_{\pi_{23},10}\approx1.4369$ and respectively set $\cand_{10}=\{\pi_{22}\}$.\\    
    \emph{Iteration 11}.
    The minimal nonzero coefficient in $\alpha_{10}^n$ is reached solely by the arrangement $\pi_{24}=\pi_{22}\cdot(10111)$.
    If $\pi\in\{\pi_{22}\cdot(01111),\pi_{22}\cdot(11111)\}$, it gives zero coefficient in $\alpha_{10}^n$ in $C_{\cP_{\pi,10}}(n)$, but a large coefficient in $\alpha_{10}^n$ in $C_{\cS_{\pi,10}}(n)$, so we exclude $\pi$.
    Thus we set $\cand_{11}=\{\pi_{24}\}$ and get $\alpha_{11}\approx1.4035$.\\    
    \emph{Iteration 12}.
    We use $(\star)$: $\{01111,11111\}$ is an isolated set and $11101$ is a worst choice, appearing in live windows only as a prefix.
    Among the coefficients in $\alpha_{11}^n$ over the functions $C_{\cP_{\pi,11}}(n)$ with $\pi$ having the prefix $\pi_{24}$, there are no zeroes and the minimum is reached at $\pi_{25}=\pi_{24}\cdot(00001)$, $\pi_{26}=\pi_{24}\cdot(01101)$, and $\pi_{27}=\pi_{24}\cdot(10001)$.
    We have $\alpha_{12}=\alpha_{\pi_{25},12}\approx1.1939
    <\alpha_{\pi_{26},12}\approx1.3247<\alpha_{\pi_{27},12}\approx1.3590$, so we set $\cand_{12}=\{\pi_{25}\}$.\\
    \emph{Iteration 13}.     
    The minimal nonzero coefficient in $\alpha_{12}^n$ is reached solely by the arrangement $\pi_{27}=\pi_{25}\cdot(01100)$.
    If $\pi\in\{\pi_{25}\cdot(00000),\pi_{25}\cdot(10000)\}$, it gives zero coefficient in $\alpha_{12}^n$ in $C_{\cP_{\pi,12}}(n)$, but a large coefficient in $\alpha_{12}^n$ in $C_{\cS_{\pi,12}}(n)$, so we exclude $\pi$.
    Thus we set $\cand_{13}=\{\pi_{27}\}$ and get $\alpha_{13}=1$, finishing the construction of the head.\\       
    \emph{Iteration 14}. 
    There are constant number of live windows.
    Among them, $0^n,1^n$, and at least one of $(011)^{\frac{n}{3}},(101)^{\frac{n}{3}},(110)^{\frac{n}{3}}$ are charged by any order. (The other live windows are obtained from the listed ones by replacing a constant number of initial letters.)
    Since $\rho$ charges, due to its tail, exactly the windows $(110)^{\frac{n}{3}}$, $0^n$, and $1^n$, its tail is optimal.
    The optimality and uniqueness of the head follows from uniqueness of the arrangement in $\cand_{13}$.
    The theorem is proved.
\end{proof}

\asymptoticdna*

\begin{proof}
    The proof is along the lines of Theorems~\ref{t:24w_opt}, \ref{t:25w_opt}.\\
    \emph{Iteration 1}. As we construct a lexmin order, Corollary~\ref{cor:cand1} implies the choice $\cand_1=\{01\}$.\\
    \emph{Iteration 2}. 
    The minimum coefficient in $\alpha_1^n$ over the functions $C_{\cP_{\pi,2}}(n)$, where $\pi[1]=01$, is reached by $\pi_1=(01,00)$, $\pi_2=(01,10)$, and $\pi_3=(01,20)$, ignoring the arrangement $(01,30)$ as non-lexmin.
    We get $\alpha_2=\alpha_{\pi_3,2}\approx3.5115<\alpha_{\pi_1,2}=\alpha_{\pi_2,2}\approx3.5616$ and set $\cand_2=\{\pi_3\}$.\\
    \emph{Iteration 3}. 
    The minimum coefficient in $\alpha_2^n$ is reached by $\pi_4=\pi_3\cdot(00)$, $\pi_5=\pi_3\cdot(10)$, and $\pi_6=\pi_3\cdot(30)$.
    We have $\alpha_3=\alpha_{\pi_6,3}\approx3.2695<\alpha_{\pi_5,3}\approx3.3028<\alpha_{\pi_4,3}\approx3.3830$ and set $\cand_3=\{\pi_6\}$.\\
    \emph{Iteration 4}. 
    The minimum coefficient in $\alpha_3^n$ is reached by $\pi_7=\pi_6\cdot(00)$ and $\pi_8=\pi_6\cdot(10)$.
    We get $\alpha_4=\alpha_{\pi_8,4}=3<\alpha_{\pi_7,4}\approx3.1958$ and set $\cand_3=\{\pi_8\}$.\\
    \emph{Iteration 5}. 
    We use $(\star)$: $00$ is a worst choice as it occurs only as a prefix of a window.
    Choosing any other 2-mer $u$, we get $C_{\cP_{\pi,5}}(n)=3^{n-2}$ for $\pi=\pi_8\cdot(u)$, so we look at $\alpha_{\pi,5}$.
    It equals $3$ for $\pi\in\{02,03\}$, $\approx2.7321$ for $\pi\in\{11,22,33\}$, and $\alpha_5\approx2.6180$ otherwise.
    By the lexmin condition, we have three options to obtain this minimal growth rate: $\pi_9=\pi_8\cdot(12)$, $\pi_{10}=\pi_8\cdot(21)$, and $\pi_{11}=\pi_8\cdot(23)$.
    We set $\cand_5=\{\pi_9,\pi_{10},\pi_{11}\}$ and observe that $C_{\cS_{\pi_{10},5}}(n)=C_{\cS_{\pi_{11},5}}(n)$, while $C_{\cS_{\pi_9,5}}(n)$ has a larger coefficient in $\alpha_5^n$.\\
    \emph{Iteration 6}. 
    The minimal coefficient in $\alpha_5^n$ is the same for the arrangements beginning with all elements of $\cand_5$; due to the larger coefficient of $C_{\cS_{\pi_9,5}}(n)$ in $\alpha_5^n$, we drop the arrangements starting with $\pi_9$.
    To get the minimum coefficient in $\alpha_5^n$, the arrangements $\pi_{10},\pi_{11}$ should be extended by some $u=a2$, where $a\in\{0,1,2,3\}$.
    We define $\pi_{12}=\pi_{10}\cdot(32)$, $\pi_{13}=\pi_{11}\cdot(12)$, as they provide the growth rate $\alpha_6=\alpha_{\pi_{12},6}=\alpha_{\pi_{13},6}\approx2.3247$, which is smaller than for the extensions by other 2-mers.
    We set $\cand_5=\{\pi_{12},\pi_{13}\}$ and observe that $C_{\cS_{\pi_{13},5}}(n)$ has larger coefficient in $\alpha_6^n$ than $C_{\cS_{\pi_{12},5}}(n)$.\\
    \emph{Iteration 7}. 
    The minimal coefficient in $\alpha_6^n$ is the same for the arrangements beginning with both elements of $\cand_6$; due to the larger coefficient of $C_{\cS_{\pi_{13},6}}(n)$ in $\alpha_6^n$, we drop the arrangements starting with $\pi_{13}$.
    To get the minimum coefficient, $\pi_{12}$ should be extended by $u=a2$ for some letter $a$.
    Comparing the growth rates for various arrangements $\pi_{12}\cdot(u)$, we get the minimum $\alpha_7=2$ for $\pi_{14}=\pi_{12}\cdot(32)$ and set $\cand_7=\{\pi_{14}\}$.\\
    \emph{Iteration 8}.  
    The minimal coefficient in $\alpha_7^n$ is reached by all extensions of $\pi_{14}$ with the 2-mers ending in 1 or 3. 
    Among them, only $\pi_{15}=\pi_{14}\cdot(31)$ and $\pi_{16}=\pi_{14}\cdot(13)$ provide the minimum possible growth rate $\alpha_8=1$.
    Next we note that $\cS_{\pi_{15},8}$ consists of windows of the form $0^i2^j3^\ell 1$ and $1^i3^\ell 1$, implying $C_{\cS_{\pi_{15},8}}(n)=\Theta(n^2)$; similarly, $\cS_{\pi_{16},8}$ consists of windows of the form $0^i2^j3^\ell1^m 3$, implying $C_{\cS_{\pi_{16},8}}(n)=\Theta(n^3)$.
    Respectively, we set $\cand_8=\{\pi_{15}\}$.\\
    \emph{Iteration 9}. 
    There is a number of possible tails that can be added to the head $\pi_{15}$ to get an optimal arrangement.
    Given that the set of live windows is $\{0^i2^j3^\ell\mid i+j+\ell=n\}\cup\{1^i3^\ell\mid i+\ell=n\}$, all these tails can be found by the reasoning as in Theorem~\ref{t:23w_opt} or by a short computer search.
    One of them completes $\pi_{15}$ to $\rho$, which is essentially unique by definition.
\end{proof}

\subsection{To Section~\ref{ss:lower}} \label{ss:lboundproof}

\lbound*

\begin{proof}
    Let $\rho\in\U_{\sigma,k}$. 
    Then $\cA_{\rho,1}$ is strongly connected, except for $\rho[1]\in\{0^{k-1}1,01^{k-1}\}$ with $\sigma=2$ (Fig.~\ref{f:aut01}).
    In the strongly connected case, $\alpha_{\rho,2}<\alpha_{\rho,1}$, as $\rho[2]$ can be read in the major (and only) component.
    Now let $E=\{\rho\mid \rho[1]\in\{0^{k-1}1,01^{k-1}\}$, and let $E_1\subset E$ consist of UHS orders with the prefix $(01^{k-1},a01^{k-2})$, where $a\in\Sigma$.
    By Lemma~\ref{l:0001}, every element of $E'$ is asymptotically better than all elements of $E-E'$, so we can discard all $\rho\in E-E'$ from the computation of $\beta_1(\sigma,k,w)$.
    For any $\rho\in E'$ we have $\alpha_{\rho,2}<\alpha_{\rho,1}$, as the vertex $\lambda.a01^{k-2}$ is in the major component of $\cA_{\rho,1}$ (Fig.~\ref{f:aut01}b).
    So we can assume $i=1$ in the definition of $B(\rho,n)$.
    Hence, $B(\rho,n)=\sigma^w+C_{\cS_{\rho,1}}(n)+C_{\cP_{\rho,2}}(n)$.
    By Lemma~\ref{l:regular}, both $C_{\cS_{\rho,1}}(n)$ and $C_{\cP_{\rho,2}}(n)$ are numbers of walks between certain vertices of $\cA_{\rho,1}$ and thus satisfy the linear recurrence having the same characteristic polynomial as $\cA_{\rho,1}$.
    Note that this polynomial is of degree $k$, because $\cA_{\rho,1}$ has $k$ vertices for any $\rho$.
    
    Let $\alpha=\min_{\rho\in\U_{\sigma,k}} \alpha_{\rho,1}$.
    Starting from some $w=w_0$, the minimum in \eqref{e:lb} is reached on some $\rho$ with $\alpha_{\rho,1}=\alpha$. 
    By Lemma~\ref{l:GuOd}, this equality holds if and only if $u=\rho[1]$ has no periods.
    To find $\alpha$, we can take arbitrary $k$-mer $u$ with no periods; we
    take $u=0^{k-1}1$ and build the canonical DFA $\cA_u$ for the language $L_u$ avoiding $u$. (See Fig.~\ref{f:aut01}a for the case $\sigma=2$; if $\sigma>2$, $(\sigma-2)$ edges from any vertex to $\lambda$ should be added.)
    Then $\alpha=g(L_u)$ is the largest eigenvalue of the matrix $A$ of $\cA_u$.
    We compute its characteristic polynomial $\chi(r)=|rI-A|=r^k-\sigma r^{k-1}+1$, expanding along the first column of $rI-A$.

    Now we localize all other zeroes of the polynomial $\chi(r)$ on the complex plane, using the classical Rouch\'e theorem.
    The theorem states that if two functions $f_1(r),f_2(r)$ are analytic inside some region $K$ in the complex plane, the contour $\partial K$ of $K$ is a simple closed curve, and $|f_1(r)|>|f_2(r)|$ for all $r\in\partial K$, then $f_1(r)$ and $f_1(r)+f_2(r)$ have the same number of zeroes (counted with multiplicities) inside $K$.

    Let $\sigma>2$ and let $K$ be the unit disk centered at the origin.
    Hence, $\partial K=\{r\in\mathbb{C}\mid |r|=1\}$.
    Taking $f_1(r)=-\sigma r^{k-1}$, $f_2(r)=r^k+1$, we have, for every $r\in\partial K$, $|f_1(r)|=\sigma>2\ge |f_2(r)|$.
    Since $f_1(r)$ has $k-1$ zeroes inside $K$, by Rouch\'e's theorem $\chi(r)=f_1(r)+f_2(r)$ also has $k-1$ zeroes inside $K$ (i.e., with absolute value $<1$).

    We claim that $\chi(r)$ is the unique monic integer polynomial of degree $k$ with the zero $\alpha$.
    Aiming at a contradiction, assume that $f(r)$ is another such polynomial; then $h(r)=\gcd(\chi(r),f(r))$ has degree $<k$ and the zero $\alpha$.
    Then $\frac{\chi(r)}{h(r)}$ is an integer polynomial with all roots being non-zero (as $\chi(0)\ne 0$) and of absolute value $<1$, which is impossible.
    Hence we conclude that an arbitrary choice of the $k$-mer $u$ with no periods would result in the same characteristic polynomial $\chi(r)$.
    Thus, if the minimum in \eqref{e:lb} for some $w\ge w_0$ is reached on $\bar\rho$, then $\cA_{\bar\rho,1}$ has the characteristic polynomial $\chi(r)$.

    Now consider the case $\sigma=2$.
    For $\varepsilon>0$, let $K_\varepsilon$ be the disk of radius $1+\varepsilon$, centred at the origin, so that $\partial K_\varepsilon=\{r\in\mathbb{C}\mid |r|=1+\varepsilon\}$.
    For $f_1(r)=-2 r^{k-1}$, $f_2(r)=r^k+1$, and every $r\in\partial K_\varepsilon$, we have $|f_1(r)|=2(1+\varepsilon)^{k-1}$ and $|f_2(r)|\le (1+\varepsilon)^k+1$.
    Then there exists some $\varepsilon_0(k)>0$ such that $|f_1(r)|>|f_2(r)|$ on $\partial K_\varepsilon$ for all $\varepsilon<\varepsilon_0(k)$.
    Since $f_1(r)$ has $k-1$ zeroes inside $K_\varepsilon$,  by Rouch\'e's theorem $\chi(r)=f_1(r)+f_2(r)$ also has $k-1$ zeroes in $K_\varepsilon$.
    As $\varepsilon$ can be arbitrarily small, $\chi(r)$ has $k-1$ zeroes of absolute value $\le1$.
    Moreover, $\sigma=2$ implies $\chi(r)=(r-1)(r^{k-1}-r^{k-2}-\ldots-r-1)$, so 1 is a simple zero of $\chi(r)$ because $\sigma=2$ implies $k>2$ by the condition of the theorem.
    Next, if $|r|=1$ and $r\ne 1$, then $|f_1(r)|=2$ and $|f_2(r)|<2$, implying $\chi(r)\ne 0$.
    Hence, $\chi(r)$ has $k-2$ zeroes of absolute value less than 1.

    By the same argument as above, $\bar\chi(r)=r^{k-1}-r^{k-2}-\ldots-r-1$ is the unique monic integer polynomial of degree $k-1$ with the zero $\alpha$. 
    Thus, if the minimum in \eqref{e:lb} for some $w\ge w_0$ is reached on $\bar\rho$, then the characteristic polynomial of $\cA_{\bar\rho,1}$ equals $\hat\chi(r)=\bar\chi(r)(r-t)$ for a constant $t$.
    As $\bar\chi(r)$ is monic, $\hat\chi(r)$ is integer only if $t$ is integer; as $|t|\le\alpha<2$, we get $t\in\{-1,0,1\}$.

    We have finally established that if the minimum in \eqref{e:lb} for some $w\ge w_0$ is reached on $\bar\rho$, then both functions $C_{\cS_{\bar\rho,1}}(n)$ and $C_{\cP_{\bar\rho,2}}(n)$ satisfy the linear recurrence with the characteristic polynomial $\chi(r)$ or $\bar\chi(r)$, which has a simple zero $\alpha>1$, 1 or 0 simple zeroes of absolute value $1$, and all other zeroes of absolute value $<1$.
    Hence every solution to this recurrence can be written as $c\alpha^n+O(1)$ (or $c\alpha^n+o(1)$ if the characteristic polynomial has no zeroes of absolute value 1 as in the case $\sigma>2$), where $c$ is a non-negative constant.
    It remains to note that $\alpha$ equals $\alpha_{\sigma,k}$ from the conditions of the theorem, and the function $C_{\cS_{\bar\rho,1}}(n)$ has a positive coefficient before $\alpha^n$ by Lemma~\ref{l:PSformulas}(i).
    The theorem is proved.
\end{proof}

\section{Scripts} \label{ss:scripts}

Below are the scripts implementing in plain Python the algorithm of Theorem~\ref{t:exhaustive} with all improvements listed in Section~\ref{s:exhaust}.
The scripts are given for the binary alphabet and can be straightforwardly adapted for bigger alphabets.
The script \verb|exhaust_w(B,k,w)| returns the domain of the first found UHS order charging less than \texttt{B} windows, or reports that \texttt{B} is optimal.
The script \verb|optuhsorder(hs,B,k,w)| returns the first found UHS order with the domain \texttt{hs}, charging less than \texttt{B} windows.
Given $k$ and $w$, the first script can be used repeatedly to obtain, via binary search starting from some upper bound, the minimum possible number of charged windows.
After that, a UHS order with this number of charged windows is recovered by the second script.

We also provide an auxiliary script \verb|density(minw,maxw,k,order)| to compute density: it reports the number of windows charged by the UHS order \texttt{order} for all $w$ from \texttt{minw} to \texttt{maxw}.
Here, \texttt{order} is a 0-based array of ranks: \texttt{order[i]} is the rank in the input UHS order of the $k$-mer representing the number \texttt{i} in binary; the $k$-mers from outside of the UHS order get the rank $2^k$.

\begin{verbatim}
def exhaust_w(B,k,w): # B = upper bound on the number of charged windows
    ## based on transition tables, fast for big w, doesn't call anything ##
    gc=[2**w]*2**(k-1)  #initialize with the number of pref-charged windows
    table = [[0,0]]*2**k  #transition table of the deBrujin automaton
    for i in range(2**k):
        a = (i<<1) & (2**k - 1) # 0-successing kmer for i
        table[i] = [a, a+1] #creating transition table (on kmers!) 
                            #for initial computation of suf-charged windows
    for suff in range(2**(k-1)): #counting suf-charged windows for rank-0 kmer
        old=[1]*2**k
        new=[0]*2**k
        for i in range(w):
            old[suff] = 0 # discard windows having i not as a suffix
            for ii in range(2**k):
                new[table[ii][0]] += old[ii]
                new[table[ii][1]] += old[ii]
                old[ii] = 0
            old, new = new, old
        gc[suff] += old[suff]
    oldlist = [(2**i,gc[i]+2) for i in range(1,2**(k-1))]
    oldlist.append((1,gc[0]+1))
    newdict = dict()
    oldsize = 2**(k-1)
    for rank in range(1, 2**k):   #main cycle
        for j in range(oldsize):
            mask, value = oldlist.pop()
            c0,c1 = rank,0
            order = [0] * 2**k
            for kmer in range(2**k):  #create an order compatible with mask
                if (mask>>kmer) & 1:
                    order[kmer] = c1
                    c1 += 1
                else:
                    order[kmer] = c0
                    c0 +=1
            for i in range(2**k):
                a = (i<<1) & (2**k - 1) # 0-successing kmer for i
                table[order[i]] = [order[a], order[a+1]] #transition table
                                                         #(on ranks!) for order
            uhs = 2**k  # a variable checking whether mask defines a UHS
            for kmer in range(2**k):
                if (mask>>kmer) & 1 == 0:    #the bit of kmer is not set in the mask
                    score = value
                    newmask = mask + 2**kmer
                    if kmer == 0 or kmer == 2**k-1:
                        score -=1 # removing correction for 0..0 or 1..1
                    old=[0]*2**k
                    new=[0]*2**k
                    old[table[order[kmer]][0]]=1 #initialization for pref
                    old[table[order[kmer]][1]]=1 #initialization for pref
                    for i in range(1,w):    #DP for pref-gc
                        for ii in range(rank, 2**k):
                            new[table[ii][0]] += old[ii]
                            new[table[ii][1]] += old[ii]
                            old[ii] = 0
                        old, new = new, old
                    score += sum(old[rank:])
                    if score < B:
                        for i in range(rank):
                            new[i] = 0 #cleaning
                            old[i] = 0 #initialization for suff
                        for i in range(rank,2**k):
                            old[i] = 1 #initialization for suff
                        for i in range(w):
                            old[order[kmer]] = 0 # discard windows with kmer  as non-suffix
                            for ii in range(rank, 2**k):
                                new[table[ii][0]] += old[ii]
                                new[table[ii][1]] += old[ii]
                                old[ii] = 0
                            old, new = new, old
                        score += old[order[kmer]]
                        if score < B:
                            vall = newdict.get(newmask)
                            if vall == None or score < vall:   #new or improved set
                                newdict[newmask] = score  #add or update newdict entry
                            if score==value: #adding kmer to mask brings no new GCs
                                uhs -= 1
                else:   # the bit of kmer is set in the mask
                    uhs -= 1
            if uhs==0:  # adding any k-mer to the mask charges no windows
                return mask,value # first found UHS order beating the upper bound
        print(rank, 'dictsize:', getsizeof(newdict))
        oldlist = list(newdict.items())
        oldsize = len(oldlist)
        if oldsize == 0:
            return rank, 'upper bound is optimal'
        newdict.clear()
        oldlist.sort(key=lambda oldlist: oldlist[1])
        print(oldsize, 'listsize:', getsizeof(oldlist), oldlist[0], oldlist[oldsize-1])
    return oldlist[0]     # the optimal score; this point is never reached
\end{verbatim}

\begin{verbatim}
def optuhsorder(hs,B,k,w): # hs = UHS domain (mask), B = bound on charged windows
    ## explicitly stores best arrangements for subsets, doesn't call anything ##
    size = hs.bit_count()
    gc=[2**w]*2**(k-1)  #initialize with the number of prefix GCs
    table = [[0,0]]*2**k  #transition table of the deBrujin automaton
    for i in range(2**k):
        a = (i<<1) & (2**k - 1) # 0-successing kmer for i
        table[i] = [a, a+1] #creating transition table (on kmers!) 
                            #for initial computation of suf-charged windows
    for suff in range(2**(k-1)): #counting suf-charged windows for rank-0 kmer
        old=[1]*2**k
        new=[0]*2**k
        for i in range(w):
            old[suff] = 0 # discard windows having i not as a suffix
            for ii in range(2**k):
                new[table[ii][0]] += old[ii]
                new[table[ii][1]] += old[ii]
                old[ii] = 0
            old, new = new, old
        gc[suff] += old[suff]
    oldlist = [(2**i,(gc[i]+1+(i>0),i)) for i in range(2**(k-1)) if (hs>>i) & 1 ==1] 
    newdict = dict()
    oldsize = len(oldlist)
    for rank in range(1, size):   #main cycle
        for j in range(oldsize):
            mask, value = oldlist.pop()
            c0,c1 = rank,0
            order = [0] * 2**k
            for kmer in range(2**k):  #create an order compatible with mask
                if (mask>>kmer) & 1:
                    order[kmer] = c1
                    c1 += 1
                else:
                    order[kmer] = c0
                    c0 +=1
            for i in range(2**k):
                a = (i<<1) & (2**k - 1) # 0-successing kmer for i
                table[order[i]] = [order[a], order[a+1]] #transition table 
                                                         #(on ranks) for order
            for kmer in range(2**k):
                if (~mask & hs)>>kmer & 1: #kmer's bit is set in hs but not in mask
                    score = value[0]
                    newmask = mask + 2**kmer
                    if kmer == 0 or kmer == 2**k-1:
                        score -=1 # removing correction for 0..0 or 1..1
                    old=[0]*2**k
                    new=[0]*2**k
                    old[table[order[kmer]][0]]=1 #initialization for pref
                    old[table[order[kmer]][1]]=1 #initialization for pref
                    for i in range(1,w):    # DP for pref-gc
                        for ii in range(rank, 2**k):
                            new[table[ii][0]] += old[ii]
                            new[table[ii][1]] += old[ii]
                            old[ii] = 0
                        old, new = new, old
                    score += sum(old[rank:])
                    if score < B:
                        for i in range(rank):
                            new[i] = 0 #cleaning
                            old[i] = 0 #initialization for suff
                        for i in range(rank,2**k):
                            old[i] = 1 #initialization for suff
                        for i in range(w):
                            old[order[kmer]] = 0 #discard windows with kmer as non-suffix
                            for ii in range(rank, 2**k):
                                new[table[ii][0]] += old[ii]
                                new[table[ii][1]] += old[ii]
                                old[ii] = 0
                            old, new = new, old
                        score += old[order[kmer]]
                        if score < B:
                            vall = newdict.get(newmask)
                            if vall == None or score < vall[0]:   #new or improved set
                                newdict[newmask] = (score, value[1]*2**k+kmer)  
                                #add or update newdict entry
        print(rank, 'dictsize:', getsizeof(newdict))
        oldlist = list(newdict.items())
        oldsize = len(oldlist)
        newdict.clear()
        oldlist.sort(key=lambda oldlist: oldlist[1])
        print(oldsize, 'listsize:', getsizeof(oldlist), oldlist[0], oldlist[oldsize-1])
    return oldlist[0]     # the optimal order
\end{verbatim}

\begin{verbatim}
def density(minw,maxw,k,order): #computes charged windows for w in [minw..maxw]
    table = [[0,0]]*2**k #transition table of the deBrujin automaton
    size = 2**k - order.count(2**k)
    c0 = size
    for kmer in range(2**k):  #completing order 
        if order[kmer]==2**k:
            order[kmer] = c0
            c0 +=1
    for i in range(2**k):
        a = (i<<1) & (2**k-1) #0-successing kmer for i
        table[order[i]] = [order[a], order[a+1]] #transition table
                                                 #(on ranks) for order
    total = [0]*maxw
    for rank in range(size):
        old=[0]*2**k
        new=[0]*2**k
        old[rank]=1 #initialization for pref
        for i in range(maxw):   #DP for pref-gc
            for ii in range(rank,2**k):
                for iii in range(2):
                    new[table[ii][iii]] += old[ii]
                old[ii] = 0
            old, new = new, old
            if i>= minw-1:
                total[i] += sum(old[rank:])
        for i in range(rank):
            new[i] = 0 #cleaning
            old[i] = 0 #initialization for suff
        for i in range(rank,2**k):
            old[i] = 1 #initialization for suff
        for i in range(maxw):
            old[rank] = 0 #discard windows with kmer of current rank as non-suffix
            for ii in range(rank,2**k):
                for iii in range(2):
                    new[table[ii][iii]] += old[ii]
                old[ii] = 0
            old, new = new, old
            if i>= minw-1:
                total[i] += old[rank]
    return total[minw-1:]
\end{verbatim}

\section{Lists of optimal UHS orders} \label{ss:listoptimal}

For reader's convenience, we mark in boldface the lowest $k$-mer distinguishing an order from the order in the previous row.

\begin{tabular}{l|l}
     $w$&Optimal UHS order ($\sigma=2$, $k=3$)\\
     \hline
     2--8& ($011,010,000,110,111$) \\
     8+& ($011,\pmb{001},101,000,110,111$)
\end{tabular}

\medskip\noindent
\begin{tabular}{l|l}
     $w$&Optimal UHS order ($\sigma=2$, $k=4$)\\
     \hline
     2& ($0101, 0001, 0100, 0111, 1101, 1100, 0011, 0000, 1111$) \\
     3--6& $(\pmb{0001},1001,1011,1110,1010,0000,1111,1101)$\\
     6,7,9,11& $(\pmb{0111},0110,0010,0000,1010,1110,1111)$\\
     6,8,10,12& $(0111,0110,0010,0000,\pmb{0101},1110,1111)$\\
     13--26& $(0111,\pmb{0001},0100,1011,1100,1010,1110,0000,1111)$\\
     27--55& $(0111,\pmb{0011},0001,0100,1100,1011,1110,0101,0000,1111)$\\
     56+& $(\pmb{0011},0001,1100,0100,1110,1011,1010,1111,0000)$\\
\end{tabular}

\medskip\noindent
\begin{tabular}{l|l}
     $w$&Optimal UHS order ($\sigma=4$, $k=2$)\\
     \hline
     2--5& ($01,31,21,32,30,10,22,33,00,12,10,11$) \\
     6--11& $(01,\pmb{32},31,01,30,21,00,33,22,12,10,11)$\\
     12--15& $(01,\pmb{31},21,32,01,30,00,33,22,12,10,11)$\\
     16,17& $(01,31,21,32,\pmb{03},23,01,33,22,13,12,10,00,11)$\\
     18& $(01,31,\pmb{23},01,03,13,21,22,32,30,12,10,00,33,11)$\\
     19--21& $(01,31,\pmb{30},10,23,22,32,10,12,33,11,01,00)$\\
     22--24& $(01,\pmb{30},31,10,10,12,32,13,33,22,03,01,11)$\\
     25+& $(01,30,\pmb{20},10,31,23,13,21,33,12,03,22,00,11)$\\
\end{tabular}

\medskip\noindent
\tabcolsep=1pt
\begin{tabular}{l|l}
     $w$&Optimal UHS order ($\sigma=2$, $k=5$)\\
     \hline
     2& \scriptsize{($01010,01000,00010,10010,01011,11010,01001,01100,01110,11011,00011,00110,01111,10011,00000,11000,11110,11111$)} \\
     3,4,9& $(\pmb{00001},10111,10011,10001,10110,10010,11110,01010,11111,00000)$\\
     4,6& $(00001,10111,10011,10001,10110,10010,\pmb{10101},11110,10100,10000,11111,00000)$\\
     4,8& $(00001,10111,10011,10001,10110,10010,10101,11110,\pmb{10000},11111,00000
)$\\
     5& $(00001,10111,10011,10001,10110,10010,\pmb{01010},11110,11010,11111,11000,00000)$\\
     7& $(00001,10111,10011,10001,10110,10010,\pmb{11110},01010,01000,00000,11111)$\\
     10& $(\pmb{01111},01100,10100,01110,00010,01101,00000,01010,10010,11110,11111)$\\
     11& $(01111,01100,10100,01110,00010,01101,00000,\pmb{10101},01001,11110,11111)$\\
     12--13& $(\pmb{01011},00111,00110,00010,10010,11110,11010,01100,00000,10111,10100,11011,10101,11111)$\\
     13--17& $(01011,00111,00110,00010,10010,11110,11010,01100,\pmb{11100},10111,00000,11011,10100,10101,11111)$\\
     18--31& $(01011,00111,00110,00010,11010,\pmb{10010},11110,11100,01100,00000,10100,10111,01010,01101,11111,01111)$\\
     32& $(01011,00111,\pmb{00001},10011,01000,11000,11110,01001,10111,01100,01010,11011,00000,11111)$\\
     33,36& $(01011,00111,00001,10011,01000,11000,11110,\pmb{10100},10010,10111,10101,10110,00000,11111)$\\
     34& $(01011,00111,00001,10011,01000,11000,11110,10100,10010,10111,10101,\pmb{11011},00000,11111)$\\
     35& $(01011,00111,00001,10011,01000,11000,11110,10100,10010,10111,\pmb{01010},11011,00000,11111)$\\
     37,38,42,43& $(\pmb{01111},00101,00011,00100,10100,10111,01100,01011,01101,00000,10011,11110,10101,11111)$\\
     39,40& $(01111,00101,00011,00100,10100,10111,01100,01011,01101,00000,\pmb{11100},11110,10101,11111)$\\
     41& $(01111,00101,00011,00100,10100,10111,01100,01011,01101,00000,\pmb{01110},11110,10101,11111)$\\
     44--48& $(01111,\pmb{01011},00011,00101,00010,10010,10100,01100,11100,10111,11011,11110,10101,00000,11111)$\\
     49--55& $(01111,01011,\pmb{00111},00011,00101,00010,10010,10100,01100,11100,10111,11011,11110,10101,00000,11111)$\\
     56,60,64& \scriptsize{$(01111,01011,00111,00011,00101,\pmb{00001},10010,10100,11000,01100,11100,10111,01101,11110,10101,00000,00100,11111)$}\\
     57,61,65& \scriptsize{$(01111,01011,00111,00011,00101,00001,10010,10100,11000,01100,11100,10111,01101,11110,10101,00000,\pmb{00010},11111)$}\\
     58,62,66& \scriptsize{$(01111,01011,00111,00011,00101,00001,10010,10100,11000,01100,11100,10111,01101,11110,10101,00000,\pmb{10001},11111)$}\\
     59,63& \scriptsize{$(01111,01011,00111,00011,00101,00001,10010,10100,11000,01100,11100,10111,01101,11110,10101,00000,\pmb{01000},11111)$}\\
     67,70,73& \scriptsize{$(01111,\pmb{00111},00011,00001,10011,01000,10100,11000,01100,11100,10111,01011,10110,01010,11110,00100,00000,11111)$}\\
     68,71& \scriptsize{$(01111,00111,00011,00001,10011,01000,10100,11000,01100,11100,10111,01011,\pmb{11011},10101,11110,00100,00000,11111)$}\\
     69,72& \scriptsize{$(01111,00111,00011,00001,10011,01000,10100,11000,01100,11100,10111,01011,\pmb{01101},10101,11110,10010,00000,11111)$}\\
     74--110& \scriptsize{$(01111,00111,00011,00001,\pmb{01000},10100,11000,10111,11010,01101,01011,10011,10010,11110,01010,00000,11111)$}\\
     111--117& \scriptsize{$(01111,00111,00011,00001,01000,\pmb{11000},10100,10111,11010,01101,01011,10011,10010,11110,01010,00000,11111)$}\\
     118,121& \scriptsize{$(\pmb{00111},00011,00001,01000,10100,11000,11100,01100,11010,11110,10111,01010,01011,01101,00100,00000,11111)$}\\
     119,122& \scriptsize{$(00111,00011,00001,01000,10100,11000,11100,01100,11010,11110,10111,\pmb{10101},10011,10110,00100,00000,11111)$}\\
     120& \scriptsize{$(00111,00011,00001,01000,10100,11000,11100,01100,11010,11110,10111,\pmb{00101},10110,00100,10101,00000,11111)$}\\
     123& \scriptsize{$(00111,00011,00001,01000,10100,11000,11100,01100,11010,11110,10111,00101,10110,00100,\pmb{01010},00000,11111)$}\\
     124--135& \scriptsize{$(\pmb{01011},00101,00010,10010,11010,11001,11100,11110,00111,10111,00001,01100,10100,01101,01010,00000,11111)$}\\
     136--223& $(01011,00101,00010,10010,\pmb{10101},11010,11001,11100,11110,00111,10111,00001,01100,11011,00000,11111)$\\
     224-261& $(01011,00101,00010,\pmb{10101},11010,10010,11001,11100,11110,00111,10111,00001,01100,11011,00000,11111)$\\
     262+& $(01011,00101,\pmb{10101},00010,11010,10010,11001,11100,11110,00111,10111,00001,01100,11011,00000,11111)$\\
\end{tabular}

\end{document}